\documentclass[sigconf]{acmart}
\usepackage[utf8]{inputenc}
\usepackage{amsthm}
\usepackage{amsmath}
\usepackage{bm}
\usepackage{thmtools,thm-restate}
\usepackage{enumitem}
\usepackage{pifont}
\usepackage{array}
\usepackage{float}
\usepackage{color}
\usepackage{caption}
\usepackage{subcaption}
\usepackage{multirow}
\usepackage{booktabs}
\usepackage{graphicx}
\usepackage{colortbl}
\usepackage{pythonhighlight}
\usepackage{hyperref}
\usepackage{hyperref}
\usepackage{comment}

\hypersetup{
colorlinks=true,
linkcolor=cyan,
citecolor=red,
urlcolor=blue}
\usepackage{caption}
\usepackage{algorithm}
\usepackage[noend]{algpseudocode}

\newtheorem{theorem}{Theorem}
\newtheorem{lemma}[theorem]{Lemma} 
\newtheorem{proposition}[theorem]{Proposition}

\newtheorem{definition}[theorem]{Definition}

\newcommand{\pluseq}{\mathrel{{+}{=}}}

\newcolumntype{P}[1]{>{\centering\arraybackslash}p{#1}}

\AtBeginDocument{%
  \providecommand\BibTeX{{%
    \normalfont B\kern-0.5em{\scshape i\kern-0.25em b}\kern-0.8em\TeX}}}

\copyrightyear{2021} 
\acmYear{2021} 
\setcopyright{acmcopyright}
\acmConference[KDD '21] {Proceedings of the 27th ACM SIGKDD Conference on Knowledge Discovery and Data Mining}{August 14--18, 2021}{Virtual Event, Singapore.}
\acmBooktitle{Proceedings of the 27th ACM SIGKDD Conference on Knowledge Discovery and Data Mining (KDD '21), August 14--18, 2021, Virtual Event, Singapore}
\acmPrice{15.00}
\acmISBN{978-1-4503-8332-5/21/08} 
\acmDOI{10.1145/XXXXXXXXXXXXX}

\begin{document}
\fancyhead{}

\title{Subset Node Representation Learning over Large Dynamic Graphs}

\author{Xingzhi Guo}
\authornote{Both authors contributed equally to this research}
\email{xingzguo@cs.stonybrook.edu}
\affiliation{%
  \institution{Stony Brook University}
  \city{Stony Brook}
  \country{USA}
}

\author{Baojian Zhou}
\authornotemark[1]
\email{baojian.zhou@cs.stonybrook.edu}
\affiliation{%
  \institution{Stony Brook University}
  \city{Stony Brook}
  \country{USA}
}

\author{Steven Skiena}
\email{skiena@cs.stonybrook.edu}
\affiliation{%
  \institution{Stony Brook University}
  \city{Stony Brook}
  \country{USA}
}

\begin{abstract}
 Dynamic graph representation learning is a task to learn node embeddings over dynamic networks, and has many important applications, including knowledge graphs, citation networks to social networks. Graphs of this type are usually large-scale but only a small subset of vertices are related in downstream tasks. Current methods are too expensive to this setting as the complexity is at best linear-dependent on both the number of nodes and edges.
 
 In this paper, we propose a new method, namely Dynamic Personalized PageRank Embedding (\textsc{DynamicPPE}) for learning a target subset of node representations over large-scale dynamic networks. Based on recent advances in local node embedding and a novel computation of dynamic personalized PageRank vector (PPV), \textsc{DynamicPPE} has two key ingredients: 1) the per-PPV complexity is $\mathcal{O}(m \bar{d} / \epsilon)$ where $m,\bar{d}$, and $\epsilon$ are the number of edges received, average degree, global precision error respectively. Thus, the per-edge event update of a single node is only dependent on $\bar{d}$ in average; and 2) by using these high quality PPVs and hash kernels, the learned embeddings have properties of both locality and global consistency. These two make it possible to capture the evolution of graph structure effectively.
 
 Experimental results demonstrate both the effectiveness and efficiency of the proposed method over large-scale dynamic networks. We apply \textsc{DynamicPPE} to capture the embedding change of Chinese cities in the Wikipedia graph during this ongoing COVID-19 pandemic \footnote{\textcolor{blue}{\url{https://en.wikipedia.org/wiki/COVID-19_pandemic}}}. Our results show that these representations successfully encode the dynamics of the Wikipedia graph.
 
\end{abstract}


\maketitle

\section{Introduction}

Graph node representation learning aims to represent nodes from graph structure data into lower dimensional vectors and has received much attention in recent years \cite{perozzi2014deepwalk,grover2016node2vec,rossi2013modeling,trivedi2018dyrep,hamilton2017inductive,kipf2017semi,hamilton2017representation}. Effective methods have been successfully applied to many real-world applications where graphs are large-scale and static \cite{ying2018graph}. However, networks such as social networks \cite{berger2006framework}, knowledge graphs \cite{ji2015knowledge}, and citation networks \cite{clement2019use} are usually time-evolving where edges and nodes are inserted or deleted over time. Computing representations of all vertices over time is prohibitively expensive because only a small subset of nodes may be interesting in a particular application. Therefore, it is important and technical challenging to efficiently learn dynamic embeddings for these large-scale dynamic networks under this typical use case.

\begin{figure}
\centering
\subfloat[Dynamic graph model]{\includegraphics[width=3.3cm]{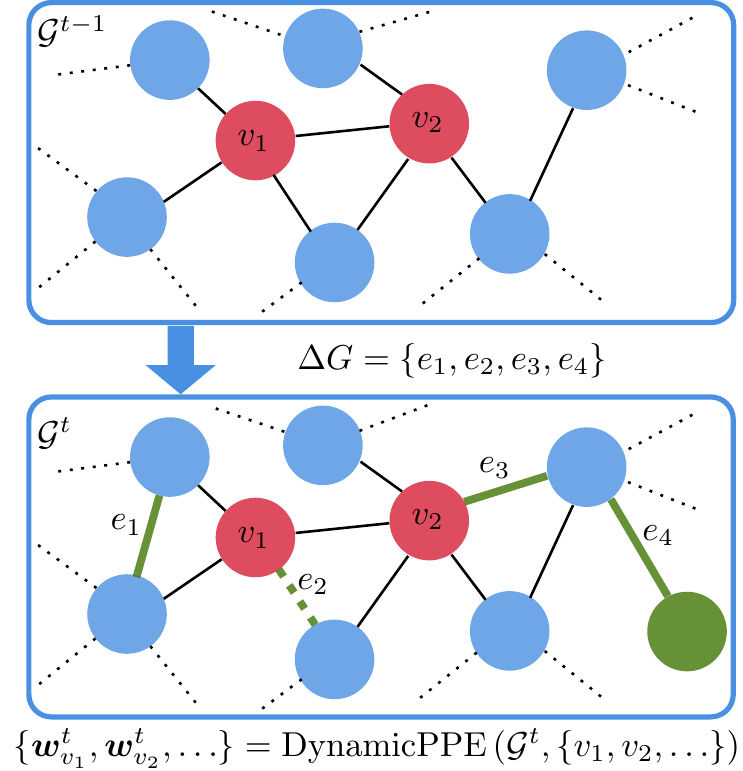}}\quad
\subfloat[\vspace{-2mm}An application of \textsc{DynamicPPE}]{\includegraphics[width=4.7cm]{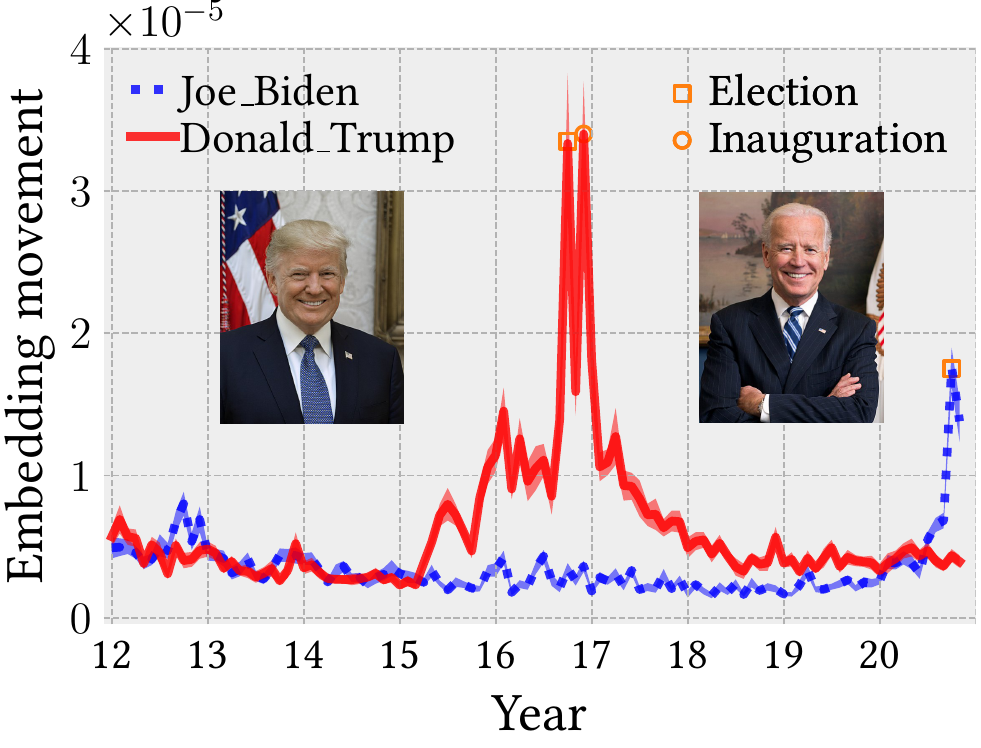}}
\caption{ (a) The model of dynamic network in two consecutive snapshots. (b) An application of \textsc{DynamicPPE} to keep track embedding movements of interesting Wikipedia articles (vertices). We learn embeddings of two presidents of the United States on the whole English Wikipedia graph from 2012 monthly, which cumulatively involves 6.2M articles (nodes) and 170M internal links (edges). The \textit{embedding movement} between two time points is defined as $1-\cos(\bm w_v^t, \bm w_v^{t+1})$ where $\cos(\cdot,\cdot)$ is the cosine similarity. The significant embedding movements may reflect big social status changes of Donald\_Trump  and Joe\_Biden \protect\footnotemark in this dynamic Wikipedia graph.\vspace{-5mm}}
\label{fig:example-model}
\end{figure}

\footnotetext{Two English Wikipedia articles are accessible at \textcolor{blue}{\url{https://en.wikipedia.org/wiki/Donald_Trump}} and \textcolor{blue}{\url{https://en.wikipedia.org/wiki/Joe_Biden}}.}

Specifically, we study the following dynamic embedding problem: We are given a subset $S=\{v_1,v_2,\ldots, v_k\}$ and an initial graph $\mathcal{G}^t$ with $t=0$. Between time $t$ and $t+1$, there are edge events of insertions and/or deletions. The task is to design an algorithm to learn embeddings for $k$ nodes with time complexity independent on the number of nodes $n$ per time $t$ where $k \ll n$. This problem setting is both technically challenging and practically important. For example, in the English Wikipedia graph, one need focus only on embedding movements of articles related to political leaders, a tiny portion of whole Wikipedia. Current dynamic embedding methods \cite{du2018dynamic,zhou2018dynamic,zhang2018billion,nguyen2018continuous,zhu2016scalable} are not applicable to this large-scale problem setting due to the lack of efficiency. More specifically, current methods have the \textit{dependence issue} where one must learn all embedding vectors. This dependence issue leads to per-embedding update is linear-dependent on $n$, which is inefficient when graphs are large-scale. This obstacle motivates us to develop a new method. 

In this paper, we propose a dynamic personalized PageRank embedding (\textsc{DynamicPPE}) method for learning a subset of node representations over large-sale dynamic networks. \textsc{DynamicPPE} is based on an effective approach to compute dynamic PPVs \cite{zhang2016approximate}. There are two challenges of using \citet{zhang2016approximate} directly: 1) the quality of dynamic PPVs depend critically on precision parameter $\epsilon$, which unfortunately is unknown under the dynamic setting; and 2) The update of per-edge event strategy is not suitable for batch update between graph snapshots. To resolve these two difficulties, first, we adaptively update $\epsilon$ so that the \textit{estimation error} is independent of $n,m$, thus obtaining high quality PPVs. Yet previous work does not give an estimation error guarantee. We prove that the time complexity is only dependent on $\bar{d}$. Second, we incorporate a batch update strategy inspired from \cite{guo2017parallel} to avoid frequent per-edge update. Therefore, the total run time to keep track of $k$ nodes for given snapshots is $\mathcal{O}(k \bar{d} m)$. Since real-world graphs have the sparsity property $\bar{d} \ll n $, it significantly improves the efficiency compared with previous methods. Inspired by \textit{InstantEmbedding} \cite{postuavaru2020instantembedding} for static graph, we use hash kernels to project dynamic PPVs into embedding space. Figure \ref{fig:example-model} shows an example of successfully applying \textsc{DynamicPPE} to study the dynamics of social status in the English Wikipedia graph. To summarize, our contributions are:

\begin{enumerate}[leftmargin=*]
\item We propose a new algorithm \textsc{DynamicPPE}, which is based on the recent advances of local network embedding on static graph and a novel computation of dynamic PPVs. \textsc{DynamicPPE} effectively learns PPVs and then projects them into embedding space using hash kernels.
\item \textsc{DynamicPPE} adaptively updates the precision parameter $\epsilon$ so that PPVs always have a provable estimation error guarantee. In our subset problem setting, we prove that the time and space complexity are all linear to the number of edges $m$ but independent on the number of nodes $n$, thus significantly improve the efficiency.
\item Node classification results demonstrate the effectiveness and efficiency of the proposed. We compile three large-scale datasets to validate our method. As an application, we showcase that learned embeddings can be used to detect the changes of Chinese cities during this ongoing COVID-19 pandemic articles on a large-scale English Wikipedia.
\end{enumerate}

The rest of this paper is organized as follows: In Section \ref{sec:related-work}, we give the overview of current dynamic embedding methods. The problem definition and preliminaries are in Section \ref{sec:defintions}. We present our proposed method in Section \ref{sec:proposed}. Experimental results are reported in Section \ref{sec:experiments}. The discussion and conclusion will be presented in Section \ref{sec:conclusion}. Our code and created datasets are accessible at \textcolor{blue}{\url{https://github.com/zjlxgxz/DynamicPPE}}.

\section{Related Work}
\label{sec:related-work}

There are two main categories of works for learning embeddings from the dynamic graph structure data. The first type is focusing on capturing the evolution of dynamics of graph structure \cite{zhou2018dynamic}. The second type is focusing on both dynamics of graph structure and features lie in these graph data \cite{trivedi2019dyrep}. In this paper, we focus on the first type and give the overview of related works. Due to the large mount of works in this area, some related works may not be included, one can find more related works in a survey \cite{kazemi2020representation} and references therein.

\noindent\textbf{Dynamic latent space models}\quad The dynamic embedding models had been initially explored by using latent space model \cite{hoff2002latent}.  The dynamic latent space model of a network makes an assumption that each node is associated with an $d$-dimensional vector and distance between two vectors should be small if there is an edge between these two nodes \cite{sarkar2005dynamic-kdd,sarkar2005dynamic-nips}. Works of these assume that the distance between two consecutive embeddings should be small. The proposed dynamic models were applied to different applications \cite{sarkar2007latent,hoff2007modeling}.  Their methods are not scalable from the fact that the time complexity of initial position estimation is at least $\mathcal{O}(n^2)$ even if the per-node update is $\log(n)$.

\noindent\textbf{Incremental SVD and random walk based methods} \quad \citet{zhang2018timers} proposed \textsc{TIMERS} that is an incremental SVD-based method. To prevent the error accumulation, \textsc{TIMERS} properly set the restart time so that the accumulated error can be reduced. \citet{nguyen2018continuous} proposed continuous-time dynamic network embeddings, namely CTDNE. The key idea of CTNDE is that instead of using general random walks as DeepWalk \cite{perozzi2014deepwalk}, it uses temporal random walks contain a sequence of edges in order. Similarly, the work of \citet{du2018dynamic} was also based on the idea of DeepWalk. These methods have time complexity dependent on $n$ for per-snapshot update. \citet{zhou2018dynamic} proposed to learn dynamic embeddings by modeling the triadic closure to capture the dynamics. 

\noindent\textbf{Graph neural network methods} \citet{trivedi2019dyrep} designed a dynamic node representation model, namely \textsc{DyRep}, as modeling a latent mediation process where it leverages the changes of node between the node's social interactions and its neighborhoods. More specifically, \textsc{DyRep} contains a temporal attention layer to capture the interactions of neighbors. \citet{zang2020neural} proposed a neural network model to learn embeddings by solving a differential equation with ReLU as its activation function. \citep{kumar2019predicting} presents a dynamic embedding, a recurrent neural network method, to learn the interactions between users and items. However, these methods either need to have features as input or cannot be applied to large-scale dynamic graph. \citet{kumar2019predicting} proposed an algorithm to learn the trajectory of the dynamic embedding for temporal interaction networks. Since the learning task is different from ours, one can find more details in their paper.

\section{Notations and preliminaries}
\label{sec:defintions}

\textbf{Notations}\quad We use $[n]$ to denote a ground set $[n]:=\{ 0,1, \ldots, n-1\}$. The graph snapshot at time $t$ is denoted as $\mathcal{G}^t\left(\mathbb{V}^t,\mathbb{E}^t\right)$. The degree of a node $v$ is $d(v)^t$. In the rest of this paper, the average degree at time $t$ is $\bar{d}^t$ and the subset of target nodes is $S\subseteq \mathbb{V}^t$. Bold capitals, e.g. $\bm A, \bm W$ are matrices and bold lower letters are vectors $\bm w,\bm x$. More specifically, the embedding vector for node $v$ at time $t$ denoted as $\bm w_v^t \in \mathbb{R}^d$ and $d$ is the embedding dimension. The $i$-th entry of $\bm w_v^t$ is  $w_v^t(i) \in \mathbb{R}$. The embedding of node $v$ for all $T$ snapshots is written as $\bm W_v = [\bm w_v^1,\bm w_v^2,\ldots, \bm w_v^T]^\top$. We use $n_t$ and $m_t$ as the number of nodes and edges in $\mathcal{G}^t$ which we simply use $n$ and $m$ if time $t$ is clear in the context.

Given the graph snapshot $\mathcal{G}^t$ and a specific node $v$, the personalized PageRank vector (PPV) is an $n$-dimensional vector $\bm \pi_v^t \in \mathbb{R}^{n}$ and the corresponding $i$-th entry is $\pi_v^t(i)$.  We use $\bm p_v^t \in \mathbb{R}^{n}$ to stand for a calculated PPV obtained from a specific algorithm. Similarly, the corresponding $i$-th entry is $p_v^t(i)$. The teleport probability of the PageRank is denoted as $\alpha$. The \textit{estimation error} of an embedding vector is the difference between true embedding $\bm w_v^t$ and the estimated embedding $\hat{\bm w}_v^t$ is measure by $\| \cdot \|_1 := \sum_{i=1}^n \left| w_v^t(i) - \hat{w}_v^t(i)\right|$.

\subsection{Dynamic graph model and its embedding}

Given any initial graph (could be an empty graph), the corresponding dynamic graph model describes how the graph structure evolves over time. We first define the dynamic graph model, which is based on \citet{kazemi2020representation}.

\begin{definition}[Simple dynamic graph model \cite{kazemi2020representation}]
\label{def:dynamic-graph-model}
A simple dynamic graph model is defined as an ordered of snapshots $\mathcal{G}^0, \mathcal{G}^1$, $\mathcal{G}^2,  \ldots , \mathcal{G}^T$ where $\mathcal{G}^0$ is the initial graph. The \textit{difference} of graph $\mathcal{G}^t$ at time $t=1,2,\ldots,T$ is $\Delta G^t(\Delta \mathbb{V}^t,\Delta\mathbb{E}^t) := G^t \backslash G^{t-1}$ with $\Delta \mathbb{V}^t := \mathbb{V}^t \backslash \mathbb{V}^{t-1}$ and $\Delta \mathbb{E}^t :=\mathbb{E}^t \backslash \mathbb{E}^{t-1}$. Equivalently, $\Delta \mathcal{G}^t$ corresponds to a sequence of edge events as the following
\begin{equation}
\Delta \mathcal{G}^t = \left\{ e_1^t,e_2^t,\ldots, e_{m^\prime}^t \right\},
\end{equation}
where each edge event $e_i^t$ has two types: \textit{insertion} or \textit{deletion}, i.e, $ e_i^t = \left(u,v, \operatorname{event} \right)$ where $\operatorname{event} \in \{ \operatorname{Insert}, \operatorname{Delete}\}$ \footnote{The node insertion can be treated as inserting a new edge and then delete it and node deletion is a sequence of deleting its edges.}.
\end{definition}

The above model captures evolution of a real-world graph naturally where the structure evolution can be treated as a sequence of edge events occurred in this graph. To simplify our analysis, we assume that the graph is undirected. Based on this, we define the subset dynamic representation problem as the following.

\begin{definition}[Subset dynamic network embedding problem]
Given a dynamic network model $\left\{\mathcal{G}^0, \mathcal{G}^1, \mathcal{G}^2, \ldots, \mathcal{G}^T\right\}$ define in Definition \ref{def:dynamic-graph-model} and a subset of target nodes $S=\left\{ v_1,v_2,\ldots,v_k \right\}$, the subset dynamic network embedding problem is to learn dynamic embeddings of $T$ snapshots for all $k$ nodes $S$ where $k \ll n$. That is, given any node $v\in S$, the goal is to learn embedding matrix for each node $v \in S$, i.e. 
\begin{equation}
\bm W_v := [\bm w_v^1,\bm w_v^2,\ldots, \bm w_v^T]^\top \text{ where }\bm w_v^t \in \mathbb{R}^d \text{ and } v \in S.
\end{equation}
\end{definition}

\subsection{Personalized PageRank}

Given any node $v$ at time $t$, the personalized PageRank vector for graph $\mathcal{G}^t$ is defined as the following

\begin{definition}[Personalized PageRank (PPR)]
\label{def:ppr}
Given normalized adjacency matrix $\bm W_t = \bm D_t^{-1} \bm A_t$ where $\bm D_t$ is a diagonal matrix with $D_t(i,i) = d(i)^t$ and $\bm A_t$ is the adjacency matrix, the PageRank vector $\bm \pi_s^t$ with respect to a source node $s$ is the solution of the following equation
\begin{equation}
\bm \pi_s^t = \alpha * \bm 1_s +  (1-\alpha)\bm \pi_s^t \bm W^t, \label{equ:ppr}
\end{equation}
where $\bm 1_s$ is the unit vector with $1_s(v)=1$ when $v =s$, 0 otherwise.
\end{definition}

There are several works on computing PPVs for static graph \cite{andersen2006local,berkhin2006bookmark,andersen2007local}. The idea is based a local push operation proposed in \cite{andersen2006local}. Interestingly, \citet{zhang2016approximate} extends this idea and proposes a novel updating strategy for calculating dynamic PPVs. We use a modified version of it as presented in Algorithm \ref{algo:forward-local-push}.

\begin{algorithm}[H]
\caption{$\textsc{ForwardPush}$ \cite{zhang2016approximate}}
\begin{algorithmic}[1]
\State \textbf{Input: }$\bm p_s, \bm r_s, \mathcal{G}, \epsilon, \beta = 0$
\While{$\exists u, r_s(u) > \epsilon d(u)$}
\State $\textsc{Push}(u)$
\EndWhile
\While{$\exists u, r_s(u) < -\epsilon d(u)$} 
\State $\textsc{Push}(u)$
\EndWhile
\State \Return $(\bm p_s, \bm r_s)$
\Procedure{Push}{$u$}
\State $p_s(u) \pluseq \alpha r_s(u)$
\For{$v \in \operatorname{Nei}(u)$}
\State $r_s(v) \pluseq (1-\alpha) r_s(u) (1-\beta) / d(u)$
\EndFor
\State $r_s(u) = (1-\alpha) r_s(u) \beta $
\EndProcedure
\end{algorithmic}
\label{algo:forward-local-push}
\end{algorithm}

Algorithm \ref{algo:forward-local-push} is a generalization from \citet{andersen2006local} and there are several variants of forward push \cite{andersen2006local,lofgren2015efficient,berkhin2006bookmark}, which are dependent on how $\beta$ is chosen (we assume $\beta=0$). The essential idea of forward push is that, at each \textsc{Push} step, the frontier node $u$ transforms her $\alpha$ residual probability $r_s(u)$ into estimation probability $p_s(u)$ and then pushes the rest residual to its neighbors. The algorithm repeats this push operation until all residuals are small enough \footnote{There are two implementation of forward push depends on how frontier is selected. One is to use a first-in-first-out (FIFO) queue \cite{gleich2015pagerank} to maintain the frontiers while the other one maintains nodes using a priority queue is used \cite{berkhin2006bookmark} so that the operation cost is $\mathcal{O}(1/\epsilon \alpha)$ instead of $\mathcal{O}(\log n/ \epsilon \alpha)$.}.  Methods based on local push operations have the following invariant property.

\begin{restatable}[Invariant property \cite{hong2016discriminating}]{lemma}{primelemma}
\label{lemma:1}
 \textsc{ForwardPush} has the following invariant property
\begin{equation}
\pi_s(u) = p_s(u) + \sum_{v \in V} r_s(v) \pi_v(u), \forall u \in \mathbb{V}.
\end{equation}
\end{restatable}

The local push algorithm can guarantee that the each entry of the estimation vector $p_s(v)$ is very close to the true value $\pi_s(v)$. We state this property as in the following

\begin{lemma}[\cite{andersen2006local,zhang2016approximate}]
\label{lemma:2}
Given any graph $\mathcal{G}(\mathbb{V},\mathbb{E})$ with $\bm p_s = \bm 0, \bm r_s = \bm 1_s$ and a constant $\epsilon$, the run time for \textsc{ForwardLocalPush} is at most $\frac{1-\|\bm r_s\|_1}{\alpha \epsilon}$ and the estimation error of $\pi_s(v)$ for each node $v$ is at most $\epsilon$, i.e. $|p_s(v) - \pi_s(v)|/d(v)| \leq \epsilon$  
\end{lemma}

The main challenge to directly use forward push algorithm to obtain high quality PPVs in our setting is that: 1) the quality $\bm p_s$ return by forward push algorithm  will critically depend on the precision parameter $\epsilon$ which unfortunately is unknown under our dynamic problem setting. Furthermore, the original update of per-edge event strategy proposed in \cite{zhang2016approximate} is not suitable for batch update between graph snapshots. \citet{guo2017parallel} propose to use a batch strategy, which is more practical in real-world scenario where there is a sequence of edge events between two consecutive snapshots. This motivates us to develop a new algorithm for dynamic PPVs and then use these PPVs to obtain high quality dynamic embedding vectors. 

\section{Proposed method}
\label{sec:proposed}

To obtain dynamic embedding vectors, the general idea is to obtain dynamic PPVs and then project these PPVs into embedding space by using two kernel functions \cite{weinberger2009feature,postuavaru2020instantembedding}.
In this section, we present our proposed method \textsc{DynamicPPE} where it contains two main components: 1) an adaptive precision strategy to control the estimation error of PPVs. We then prove that the time complexity of this dynamic strategy is still independent on $n$. With this quality guarantee, learned PPVs will be used as proximity vectors and be "projected" into lower dimensional space based on ideas of \textit{Verse} \cite{tsitsulin2018verse} and \textit{InstantEmbedding} \cite{postuavaru2020instantembedding}. We first show how can we get high quality PPVs and then present how use PPVs to obtain dynamic embeddings. Finally, we give the complexity analysis.

\subsection{Dynamic graph embedding for single batch}

For each batch update $\Delta G^t$, the key idea is to dynamically maintain PPVs where the algorithm updates the estimate from $\bm p_v^{t-1}$ to $\bm p_v^t$ and its residuals from $\bm r_v^{t-1}$ to $\bm r_v^t$. Our method is inspired from \citet{guo2017parallel} where they proposed to update a personalized contribution vector by using the local reverse push \footnote{One should notice that, for undirected graph, PPVs can be calculated by using the invariant property from the contribution vectors. However, the invariant property does not hold for directed graph. It means that one cannot use reverse local push to get a personalized PageRank vector directly. In this sense, using forward push algorithm is more general for our problem setting.}. The proposed dynamic single node embedding, \textsc{DynamicSNE} is illustrated in Algorithm \ref{algo:dynamic-sne}. It takes an update batch $\Delta \mathcal{G}^t$ (a sequence of edge events), a target node $s$ with a precision $\epsilon^t$, estimation vector of $s$ and residual vector as inputs. It then obtains an updated embedding vector of $s$ by the following three steps: 1) It first updates the estimate vector $\bm p_s^t$ and $\bm r_s^t$ from Line 2 to Line 9; 2) It then calls the forward local push method to obtain the updated estimations, $\bm p_s^t$; 3) We then use the hash kernel projection step to get an updated embedding. This projection step is from \textit{InstantEmbedding} where two universal hash functions are defined as $h_d: \mathbb{N} \rightarrow [d]$ and $h_{\operatorname{sgn}} : \mathbb{N} \rightarrow \{\pm 1\}$ \footnote{For example, in our implementation, we use MurmurHash \textcolor{blue}{\url{https://github.com/aappleby/smhasher}}}. Then the hash kernel based on these two hash functions is defined as $H_{h_{\operatorname{sgn}},h_d}(\bm x): \mathbb{R}^n \rightarrow \mathbb{R}^d$ where each entity $i$ is $\sum_{j \in h_d^{-1}(i)} x_j h_{\operatorname{sgn}}(j)$. Different from random projection used in RandNE \cite{zhang2018billion} and FastRP \cite{chen2019fast},  hash functions has $\mathcal{O}(1)$ memory cost while random projection based method has $\mathcal{O}(d n)$ if the Gaussian matrix is used. Furthermore, hash kernel keeps unbiased estimator for the inner product \cite{weinberger2009feature}.

In the rest of this section, we show that the time complexity is $\mathcal{O}(m \bar{d} / \epsilon)$ in average and the estimation error of learned PPVs measure by $\| \cdot \|_1$ can also be bounded. Our proof is based on the following lemma which follows from \citet{zhang2016approximate,guo2017parallel}.

\begin{lemma}
\label{lemma:6}
Given current graph $\mathcal{G}^t$ and an update batch $\Delta \mathcal{G}^t$, the total run time of the dynamic single node embedding, \textsc{DynamicSNE} for obtaining embedding $\bm w_s^t$ is bounded by the following
\begin{equation}
T_t \leq \frac{\|\bm r_s^{t-1}\|_1 - \|\bm r_s^t\|_1}{\alpha \epsilon^t} + \sum_{u\in \Delta \mathcal{G}^t}\frac{2-\alpha}{\alpha} \frac{p_s^{t-1}(u)}{d(u)}
\end{equation}
\end{lemma}
\begin{proof}
We first make an assumption that there is only one edge update $(u,v,\operatorname{event})$ in $\Delta \mathcal{G}^t$, then based Lemma 14 of \cite{zhang2016approximate-extended}, the run time of per-edge update is at most:
\begin{equation}
\frac{\|\bm r_s^{t-1}\|_1 - \|\bm r_s^{t-1}\|_1}{\alpha \epsilon^t} + \frac{\Delta_t(s)}{\alpha \epsilon^t},
\end{equation}
where $\Delta_t(s) = \frac{2-\alpha}{\alpha} \frac{p_s^{t-1}(u)}{d(u)}$. Suppose there are $k$ edge events in $\Delta \mathcal{G}^t$. We still obtain a similar bound, by the fact that forward push algorithm has monotonicity property: the entries of estimates $ p_s^t(v)$ only increase when it pushes positive residuals (Line 2 and 3 of Algorithm \ref{algo:forward-local-push}). Similarly, estimates $p_s^t(v)$ only decrease when it pushes negative residuals (Line 4 and 5 of Algorithm \ref{algo:forward-local-push}). In other words, the amount of work for per-edge update is not less than the amount of work for per-batch update. Similar idea is also used in \cite{guo2017parallel}.
\end{proof}

\newcommand{\eqplus}{=\mathrel{+}}
\begin{algorithm}[H]
\caption{$\textsc{DynamicSNE}(\mathcal{G}^t, \Delta \mathcal{G}^t,  s,\bm p_s^{t-1},\bm r_s^{t-1}, \epsilon^t,\alpha)$}
\begin{algorithmic}[1]
\State \textbf{Input: } graph $\mathcal{G}^t, \Delta \mathcal{G}^t$, target node $s$, precision $\epsilon$, teleport $\alpha$
\For{$(u,v,\operatorname{op}) \in \Delta G^{t}$}
\If{$\operatorname{op}==\textsc{Insert}(u,v)$}
\State $\Delta_p = p_s^{t-1}(u) / (d(u)^t -1)$
\EndIf
\If{$\operatorname{op} == \textsc{Delete}(u,v)$}
\State $\Delta_p = -p_s^{t-1}(u) / (d(u)^t + 1)$
\EndIf
\State $p_s^{t-1}(u) \leftarrow p_s^{t-1}(u) + \Delta_p$
\State $r_s^{t-1}(u) \leftarrow r_s^{t-1}(u) - \Delta_p / \alpha$
\State $r_s^{t-1}(v) \leftarrow r_s^{t-1}(v) + \Delta_p / \alpha - \Delta_p$
\EndFor
\State $\bm p_s^t = \textsc{ForwardPush}(\bm p_s^{t-1},\bm r_s^{t-1},\mathcal{G}^t,\epsilon^t,\alpha)$
\State $\bm w_s^t = \bm 0$
\For{ $i \in \{ v : p_s^t(v) \ne 0, v \in \mathbb{V}^t\}$} 
\State $w_s^{t}(h_d(i)) \pluseq h_{\operatorname{sgn}}(i) \max\left(\log \left(p_s^t(i) n^t\right), 0\right)$
\EndFor
\end{algorithmic}
\label{algo:dynamic-sne}
\end{algorithm}

\begin{theorem}
Given any graph snapshot $\mathcal{G}^t$ and an update batch $\Delta G^t$ where there are $m_t$ edge events and suppose the precision parameter is $\epsilon^t$ and teleport probability is $\alpha$, \textsc{DynamicSNE} runs in $\mathcal{O}(m_t/\alpha^2 + m_t \bar{d}^t / (\epsilon\alpha^2) + m_t / (\epsilon\alpha))$ with $\epsilon^t = \epsilon / m_t$
\label{thm:run-time}
\end{theorem}
\begin{proof}
Based lemma \ref{lemma:6}, the proof directly follows from Theorem 12 of \cite{zhang2016approximate-extended}.
\end{proof}

The above theorem has an important difference from the previous one \cite{zhang2016approximate}. We require that the precision parameter will be small enough so that $\| \bm p_s^t - \bm \pi_s^t\|_1$ can be bound ( will be discussed later). As we did the experiments in Figure \ref{fig:epsilon-fix}, the fixed epsilon will make the embedding vector bad. We propose to use a dynamic precision parameter, where $\epsilon^t \sim \mathcal{O}(\epsilon^\prime /m_t)$, so that the $\ell_1$-norm can be properly bounded. Interestingly, this high precision parameter strategy will not case too high time complexity cost from $\mathcal{O}(m/\alpha^2)$ to $\mathcal{O}(m \bar{d} / (\epsilon\alpha^2))$. The bounded estimation error is presented in the following theorem.

\begin{theorem}[Estimation error]
Given any node $s$, define the estimation error of PPVs learned from \textsc{DynamicSNE} at time $t$ as $\| \bm p_s^t - \bm \pi_s^t \|_1$, if we are given the precision parameter $\epsilon^t = \epsilon / m_t$, the estimation error can be bounded by the following
\begin{equation}
\| \bm p_s^t - \bm \pi_s^t \|_1 \leq \epsilon,
\end{equation}
where we require $\epsilon \leq 2$ \footnote{By noticing that $\|\bm p_s^t - \bm \pi_s^t\|_1 \leq \|\bm p_s^t\|_1 + \| \bm \pi_s^t\|_1 \leq 2$, any precision parameter larger than 2 will be meaningless. } and $\epsilon$ is a global precision parameter of \textsc{DynamicPPE} independent on $m_t$ and $n_t$.
\end{theorem}
\begin{proof}
Notice that for any node $u$, by Lemma \ref{lemma:2}, we have the following inequality
\begin{equation}
| \pi_s^t(u) - \pi_s^t(u) | \leq \epsilon d^t(u). \nonumber
\end{equation}
Summing all these inequalities over all nodes $u$, we have
\begin{align*}
\| \bm p_s^t - \bm \pi_s^t \|_1 &= \sum_{u\in \mathbb{V}^t} \left| p_s^t(u) - \pi_s^t(u)\right| \\
&\leq \sum_{u \in \mathbb{V}^t} \epsilon^t d^t(u) = \epsilon^t m_t = \frac{\epsilon}{m_t} m_t = \epsilon.
\end{align*}
\end{proof}

\begin{figure}[H]
\centering
\includegraphics[width=8cm,height=3.5cm]{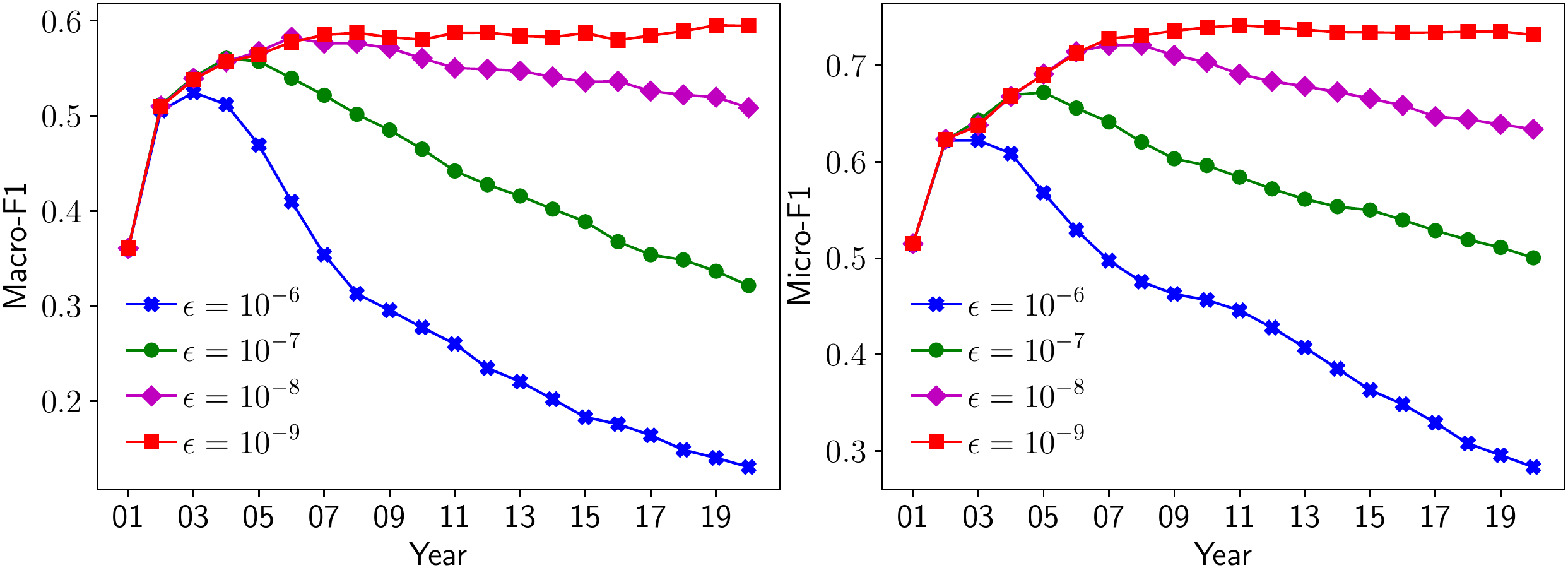}
\caption{$\epsilon$ as a function of year for the task of node classification on the English Wikipedia graph. Each line corresponds to a fixed precision strategy of \textsc{DynamicSNE}. Clearly, when the precision parameter $\epsilon$ decreases, the performance of node classification improves.}
\label{fig:epsilon-fix}
\end{figure}

The above theorem gives estimation error guarantee of $\bm p_s^t$, which is critically important for learning high quality embeddings. First of all, the dynamic precision strategy is inevitable because the precision is unknown parameter for dynamic graph where the number of nodes and edges could increase dramatically over time. To demonstrate this issue, we conduct an experiments on the English Wikipedia graph where we learn embeddings over years and validate these embeddings by using node classification task. As shown in Figure \ref{fig:epsilon-fix}, when we use the fixed parameter, the performance of node classification is getting worse when the graph is getting bigger. This is mainly due to the lower quality of PPVs. Fortunately, the adaptive precision parameter $\epsilon / m_t$ does not make the run time increase dramatically. It only dependents on the average degree ${\bar d}^t$. In practice, we found $\epsilon =0.1$ are sufficient for learning effective embeddings.

\subsection{\textsc{DynamicPPE}}

Our finally algorithm \textsc{DynamicPPE} is presented in Algorithm \ref{algo:dynamic-ppe}. At every beginning, estimators $\bm p_s^t$ are set to be zero vectors and residual vectors $\bm r_s^t$ are set to be unit vectors with mass all on one entry (Line 4). The algorithm then call the procedure \textsc{DynamicSNE} with an empty batch as input to get initial PPVs for all target nodes \footnote{For the situation that some nodes of $S$ has not appeared in $\mathcal{G}^t$ yet, it checks every batch update until all nodes are initialized. } (Line 5). From Line 6 to Line 9, at each snapshot $t$, it gets an update batch $\Delta \mathcal{G}^t$ at Line 7 and then calls \textsc{DynamicSNE} to obtain the updated embeddings for every node $v$.

\begin{algorithm}[H]
\caption{$\textsc{DynamicPPE}(\mathcal{G}_0, S,\epsilon,\alpha)$}
\begin{algorithmic}[1]
\State \textbf{Input: } initial graph $\mathcal{G}^0$, target set $S$, global precision $\epsilon$, teleport probability $\alpha$
\State $t=0$
\For{$s \in S := \{v_1,v_2,\ldots, v_k\}$}
\State $\bm p_s^t = \bm 0,\quad \bm r_s^t = \bm 1_s$
\State \textsc{DynamicSNE}$(\mathcal{G}^0, \emptyset, s, \bm p_s^t, \bm r_s^t, 1/ m_0, \alpha)$
\EndFor
\For{$t \in \{1,2,\ldots, T\}$}
\State \text{ read a sequence of edge events $\Delta \mathcal{G}^t := \mathcal{G}^t \backslash \mathcal{G}^{t-1}$}
\For{$s \in S := \{v_1,v_2,\ldots, v_k\}$}
\State $ \bm w_s^t = \textsc{DynamicSNE}(\mathcal{G}^t, \Delta \mathcal{G}^t, s, \bm p_s^{t-1}, \bm r_s^{t-1}, \epsilon/m_t,\alpha)$
\EndFor
\EndFor
\State \textbf{return} $\bm W_s^t = [\bm w_s^1,\bm w_s^2,\ldots, \bm w_s^T], \forall s\in S$. 
\end{algorithmic}
\label{algo:dynamic-ppe}
\end{algorithm}

Based on our analysis, \textsc{DynamicPPE} is an dynamic version of \textit{InstantEmbedding}. Therefore, \textsc{DynamicPPE} has two key properties observed in \cite{postuavaru2020instantembedding}: \textit{locality} and \textit{global consistency}. The embedding quality is guaranteed from the fact that \textit{InstantEmbedding} implicitly factorizes the proximity matrix based on PPVs \cite{tsitsulin2018verse}.

\subsection{Complexity analysis}

\textbf{Time complexity}\quad The overall time complexity of \textsc{DynamicPPE} is the $k$ times of the run time of \textsc{DynamicSNE}. We summarize the time complexity of \textsc{DynamicPPE} as in the following theorem
\begin{theorem} 
The time complexity of \textsc{DynamicPPE} for learning a subset of $k$ nodes is $\mathcal{O}(k \frac{m_t}{\alpha^2} + k \frac{m_t \bar{d}^t}{\alpha^2} + \frac{m_t}{\epsilon} + k T \min\{n, \frac{m}{\epsilon \alpha}\})$
\end{theorem}

\begin{proof}
We follow Theorem \ref{thm:run-time} and summarize all run time together to get the final time complexity.
\end{proof}

\noindent\textbf{Space complexity}\quad The overall space complexity has two parts: 1) $\mathcal{O}(m)$ to store the graph structure information; and 2) the storage of keeping nonzeros of $\bm p_s^t$ and $\bm r_s^t$. From the fact that local push operation \cite{andersen2006local}, the number of nonzeros in $\bm p_s^t$ is at most $\frac{1}{\epsilon \alpha}$. Thus, the total storage for saving these vectors are $\mathcal{O}(k \min\{n,\frac{m }{\epsilon \alpha}\})$. Therefore, the total space complexity is $\mathcal{O}(m + k \min(n,\frac{m}{\epsilon\alpha}))$.

\noindent\textbf{Implementation}\quad Since learning the dynamic node embedding for any node $v$ is independent with each other, \textsc{DynamicPPE} is  are easy to parallel. Our current implementation can take advantage of multi-cores and compute the embeddings of $S$ in parallel.

\section{Experiments}
\label{sec:experiments}

To demonstrate the effectiveness and efficiency of \textsc{DynamicPPE}, in this section, we first conduct experiments on several small and large scale real-world dynamic graphs on the task of node classification, followed by a case study about changes of Chinese cities in Wikipedia graph during the COVID-19 pandemic.

\subsection{Datasets}

We compile the following three real-world dynamic graphs, more details can be found in Appendix \ref{appendix::data}. 

\noindent \textbf{Enwiki20} {English Wikipedia Network\quad} We collect the internal Wikipedia Links (WikiLinks) of English Wikipedia from the beginning of Wikipedia, January 11th, 2001, to December 31, 2020 \footnote{We collect the data from the dump \url{https://dumps.wikimedia.org/enwiki/20210101/}}. The internal links are extracted using a regular expression proposed in \cite{consonni2019wikilinkgraphs}. During the entire period, we collection 6,151,779 valid articles\footnote{A valid Wikipedia article must be in the 0 namespace}. We generated the WikiLink graphs only containing edge insertion events. We keep all the edges existing before Dec. 31 2020, and sort the edge insertion order by the creation time. There are  6,216,199 nodes and 177,862,656 edges during the entire period. Each node either has one label (\textit{Food, Person, Place,...}) or no label.

\noindent \textbf{Patent} ({US Patent Graph}) \quad The citation network of US patent\cite{hall2001nber} contains 2,738,011 nodes with 13,960,811 citations range from year 1963 to 1999. For any two patents $u$ and $v$, there is an edge $(u,v)$ if the patent $u$ cites patent $v$. Each patent belongs to six different types of patent categories. We extract a small weakly-connected component of 46,753 nodes and 425,732 edges with timestamp, called \textit{Patent-small}.

\noindent \textbf{Coauthor} \quad We extracted the co-authorship network from the Microsoft Academic Graph \cite{sinha2015overview} dumped on September 21, 2019. We collect the papers with less than six coauthors, keeping those who has more than 10 publications, then build undirected edges between each coauthor with a timestamp of the publication date. In addition, we gather temporal label (e.g.: \textit{Computer Science, Art, ... }) of authors based on their publication's field of study. Specifically we assign the label of an author to be the field of study where s/he published the most up to that date. The original graph contains 4,894,639 authors and 26,894,397 edges ranging from year 1800 to 2019. We also sampled a small connected component containing 49,767 nodes and 755,446 edges with timestamp, called \textit{Coauthor-small}.

\noindent \textbf{Academic} The co-authorship network is from the academic network \cite{zhou2018dynamic,tang2008arnetminer} where it contains 51,060 nodes and 794,552 edges. The nodes are generated from 1980 to 2015. According to the node classification setting in \citet{zhou2018dynamic}, each node has either one binary label or no label.

\subsection{Node Classification Task}
\noindent\textbf{Experimental settings\quad}
We evaluate embedding quality on binary classification for Academic graph (as same as in \cite{zhou2018dynamic}), while using multi-class classification for other tasks. We use balanced logistic regression with $\ell_2$ penalty in on-vs-rest setting, and report the Macro-F1 and Macro-AUC (ROC) from 5 repeated trials with 10\% training ratio on the labeled nodes in each snapshot, excluding dangling nodes. Between each snapshot, we insert new edges and keep the previous edges.

We conduct the experiments on the aforementioned small and large scale graph. 
In small scale graphs, we calculate the embeddings of all nodes and compare our proposed method (\textbf{DynPPE.}) against other state-of-the-art models from three categories\footnote{Appendix \ref{appendix::parameter} shows the hyper-parameter settings of baseline methods}.
1) Random walk based static graph embeddings (\textbf{Deepwalk}\footnote{\url{https://pypi.org/project/deepwalk/}} \cite{perozzi2014deepwalk}, \textbf{Node2Vec}\footnote{\url{https://github.com/aditya-grover/node2vec}} \cite{grover2016node2vec}); 
2) Random Projection-based embedding method which supports online update: \textbf{RandNE}\footnote{\url{https://github.com/ZW-ZHANG/RandNE/tree/master/Python}} \cite{zhang2018billion}; 
3) Dynamic graph embedding method: \textbf{DynamicTriad} (\textbf{DynTri.}) \footnote{\url{https://github.com/luckiezhou/DynamicTriad}} \cite{zhou2018dynamic}  which models the dynamics of the graph and optimized on link prediction.

\begin{table}[!h]
\centering
\caption{Node classification task on the \textit{Academic, Patent Small, Coauthor Small} graph on the final snapshot}
\small
\begin{tabular}{p{0.05\textwidth}|p{0.06\textwidth}|p{0.03\textwidth}|p{0.03\textwidth}|p{0.03\textwidth}|p{0.03\textwidth}|p{0.03\textwidth}|p{0.03\textwidth} }
\toprule
  \multicolumn{2}{l|}{}   & \multicolumn{2}{l|}{Academic} & \multicolumn{2}{l|}{Patent Small} & \multicolumn{2}{l}{Coauthor Small}\\
  \multicolumn{2}{l|}{} &  F1 &  AUC  &  F1 &  AUC  &  F1 &  AUC \\
 \cline{1-8}
\multicolumn{8}{c}{$d=128$} \\
\cline{1-8}
  \multirow{2}{0mm}{Static method} & Node2Vec &0.833 &\textbf{0.975} &0.648 &0.917 &0.477 &\textbf{0.955}\\
\cline{2-8}
  &Deepwalk &\textbf{0.834} &\textbf{0.975} &\textbf{0.650} &\textbf{0.919} &\textbf{0.476} &0.950\\
\cline{1-8}
  \multirow{4}{0mm}{Dynamic method}   &DynTri. &0.817 &0.969 &0.560 &0.866 &0.435 &0.943\\
\cline{2-8}
    &RandNE &0.587 &0.867 &0.428 &0.756 &0.337 &0.830\\
\cline{2-8}
    &DynPPE. &0.808 &0.962 &0.630 &0.911 &0.448 &0.951\\
\hline
    \multicolumn{8}{c}{$d=512$} \\
\hline
  \multirow{2}{0mm}{Static method} & Node2Vec &0.841 &\textbf{0.975} &0.677 &0.931 &0.486 &0.955 \\
\cline{2-8}
  &Deepwalk &\textbf{0.842} &\textbf{0.975} &0.680 &0.931 &0.495 &0.955 \\
\cline{1-8}
  \multirow{4}{0mm}{Dynamic method}   &DynTri. &0.811 &0.965 &0.659 &0.915 &0.492 &0.952 \\
\cline{2-8}
    &RandNE &0.722 &0.942 &0.560 &0.858 &0.493 &0.895 \\
\cline{2-8}
    &DynPPE. &\textbf{0.842} &0.973 &\textbf{0.682} & \textbf{0.934} &\textbf{0.509} &\textbf{0.958}\\
 \bottomrule
\end{tabular}
\label{tab:nc-small-table}
\end{table}

\begin{table}[!h]
\caption{Total CPU time for small graphs (in second). RandNE-I is with orthogonal projection, the performance is slightly better, but the running time is significantly increased. RandNE-II is without orthogonal projection.}

\centering
\begin{tabular}{p{0.08\textwidth}|p{0.08\textwidth}|p{0.12\textwidth}|p{0.12\textwidth}}
\toprule
  & Academic & Patent Small & Coauthor Small \\
\hline 
Deepwalk\footnotemark & 498211.75& 181865.56 & 211684.86 \\
\hline 
Node2vec & 4584618.79& 2031090.75 & 1660984.42 \\
\hline 
 DynTri. & 247237.55 & 117993.36 & 108279.4 \\
\hline 
 RandNE-I & 12732.64 & 9637.15 & 8436.79 \\
\hline 
 RandNE-II &  1583.08 & 9208.03 & 177.89 \\
\hline 
 DynPPE. & 18419.10 & 3651.59 & 21323.74 \\
\bottomrule
\end{tabular}
\label{tab:runtime-small}
\end{table}
\footnotetext{ We ran static graph embedding methods over a set of sampled snapshots and estimate the total CPU time for all snapshots.}

\begin{figure*}[h!]
	\centering
	\begin{subfigure}{.33\textwidth}
		\includegraphics[width=5cm]{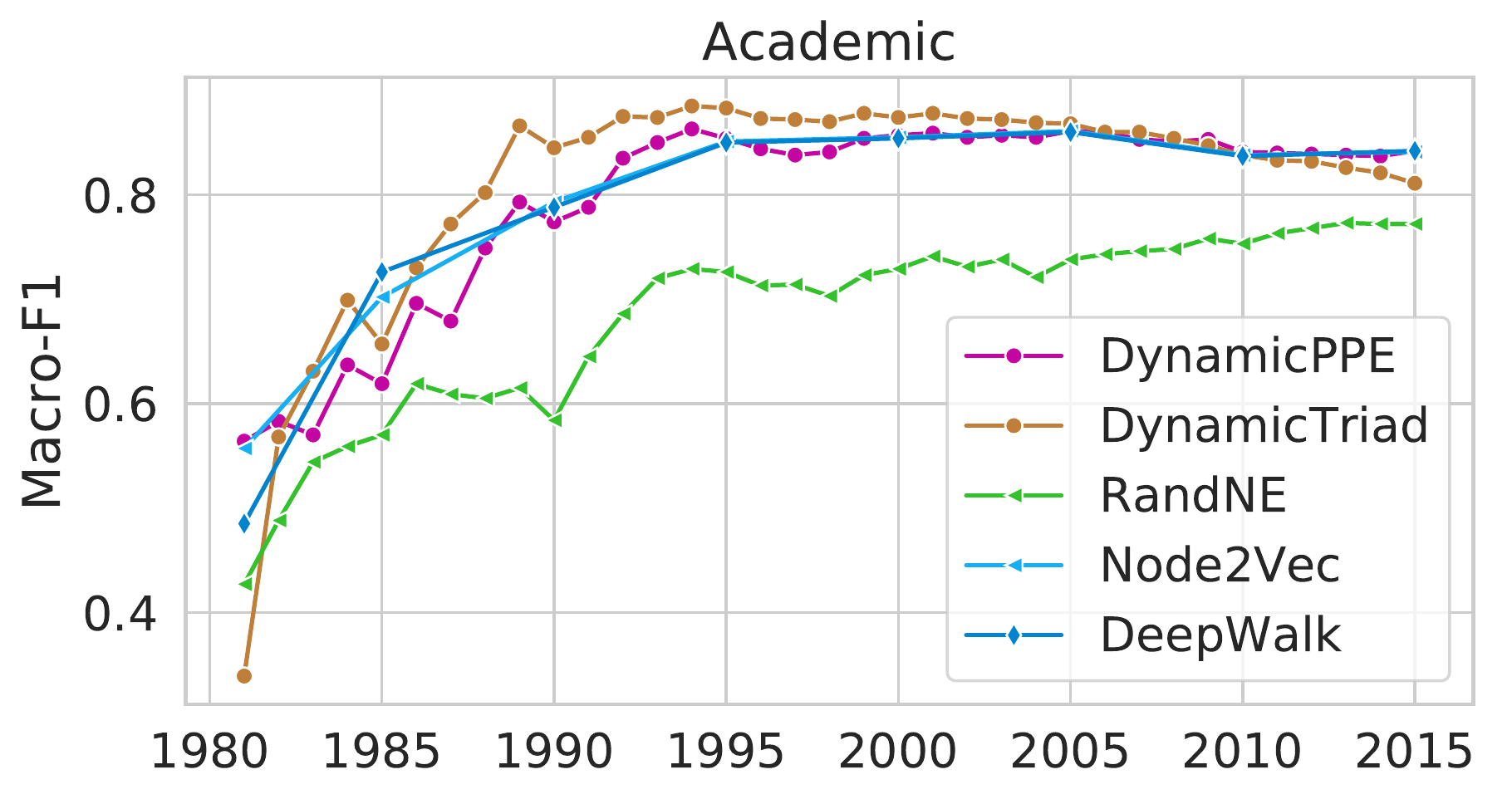}%
	\end{subfigure}
		\begin{subfigure}{.33\textwidth}
		\includegraphics[width=5cm]{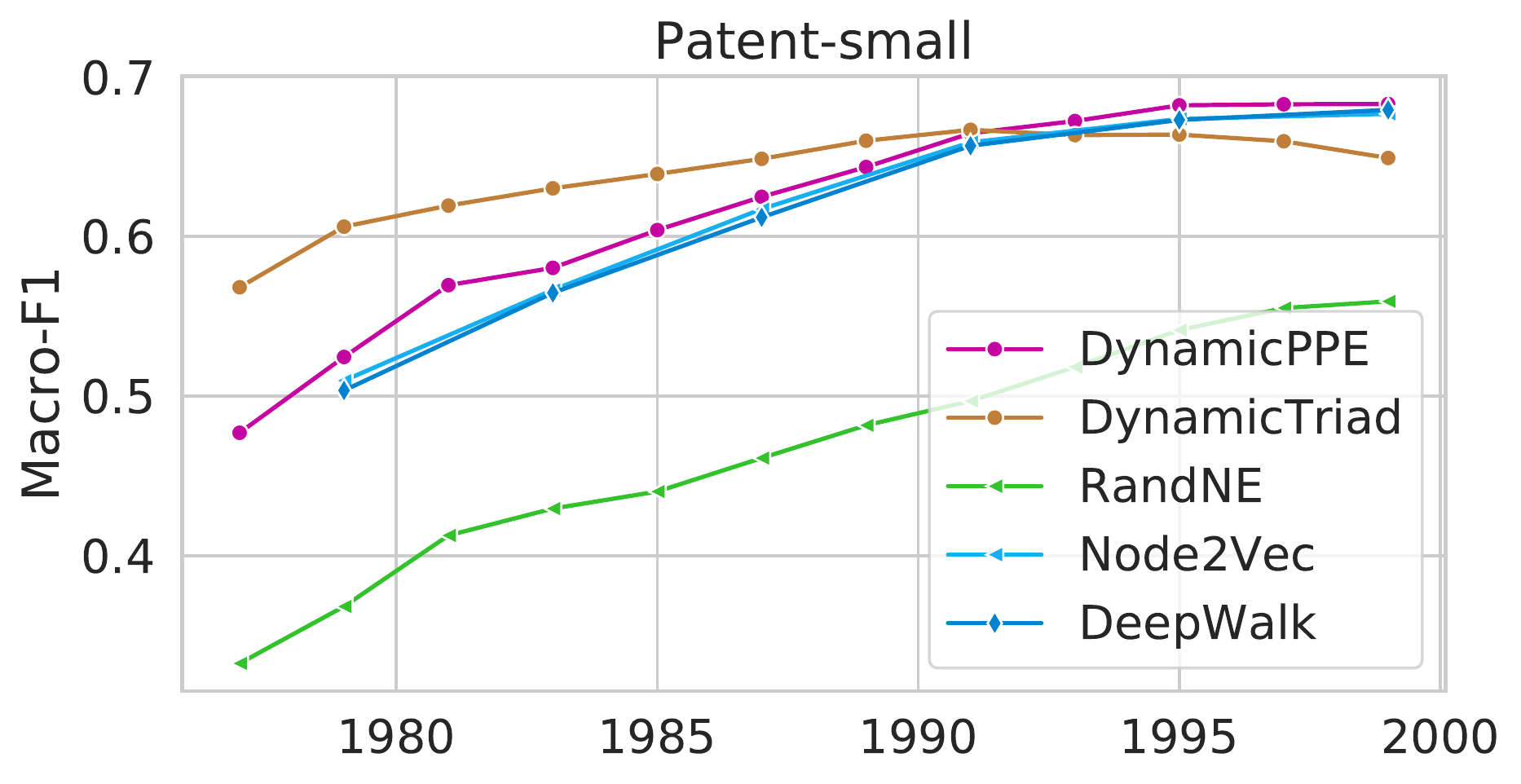}%
	\end{subfigure}
	\begin{subfigure}{.33\textwidth}
		\includegraphics[width=5cm]{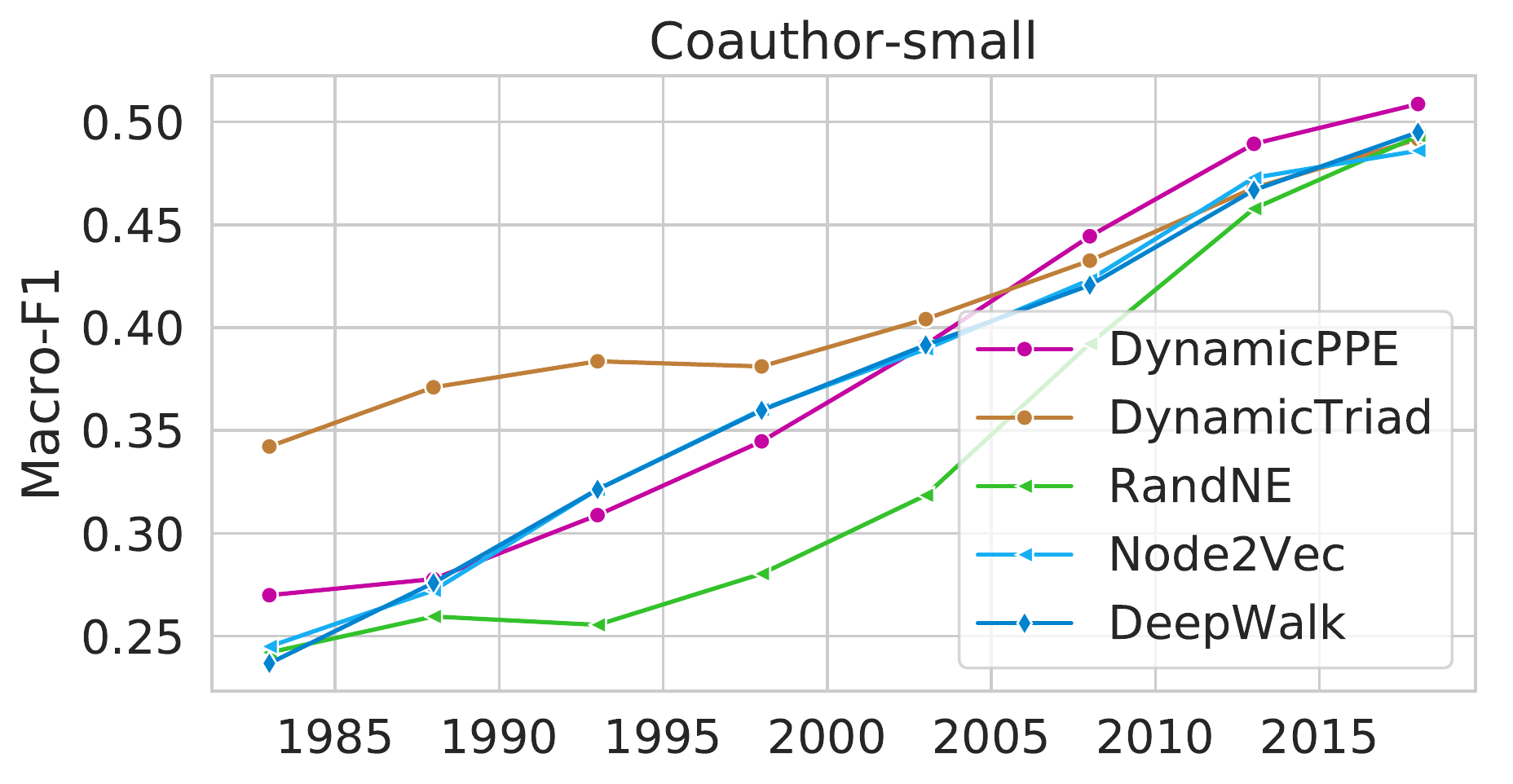}%
	\end{subfigure}
	
	\begin{subfigure}{.33\textwidth}
		\includegraphics[width=5cm]{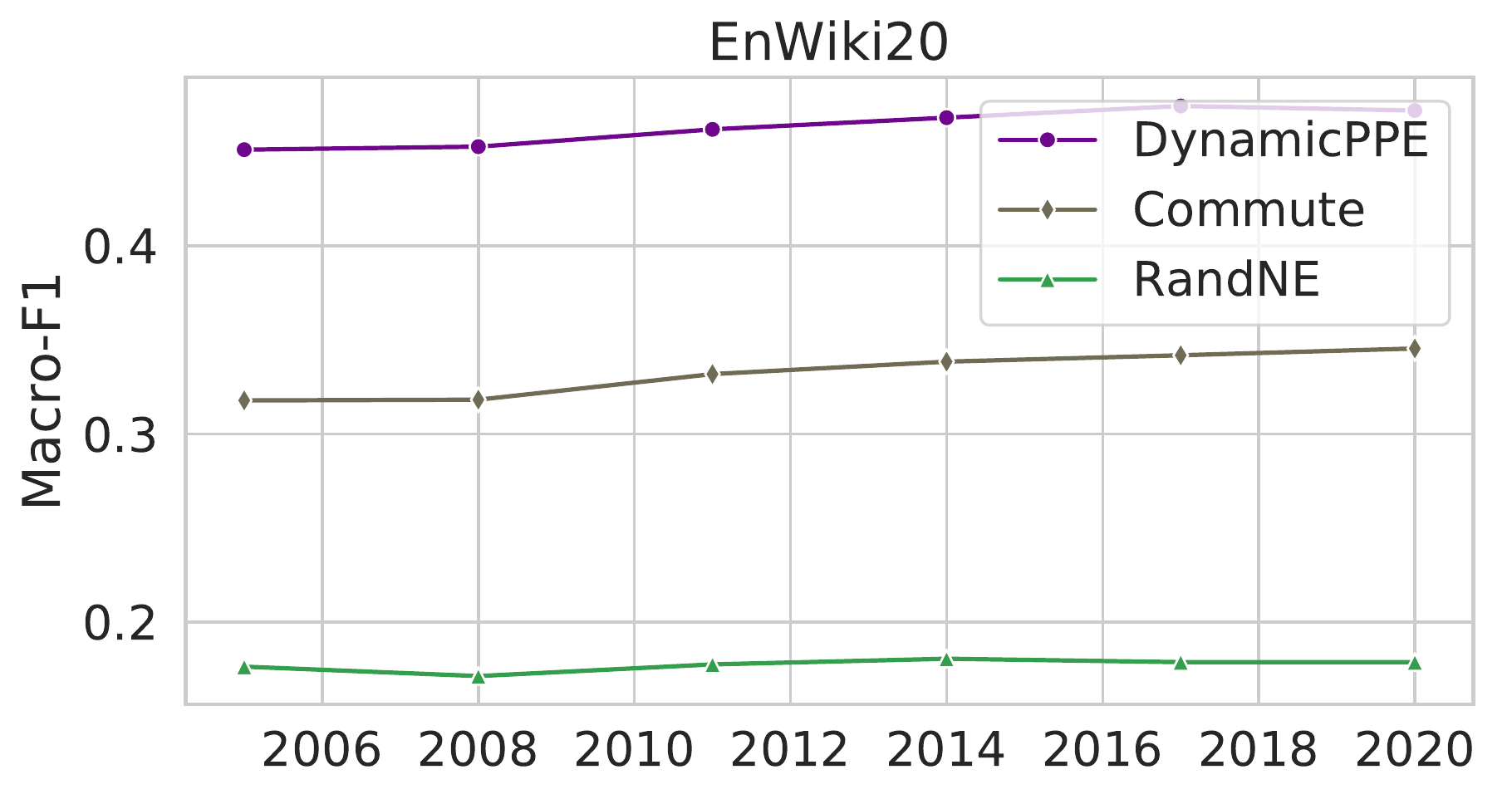}%
	\end{subfigure}
	\begin{subfigure}{.33\textwidth}
		\includegraphics[width=5cm]{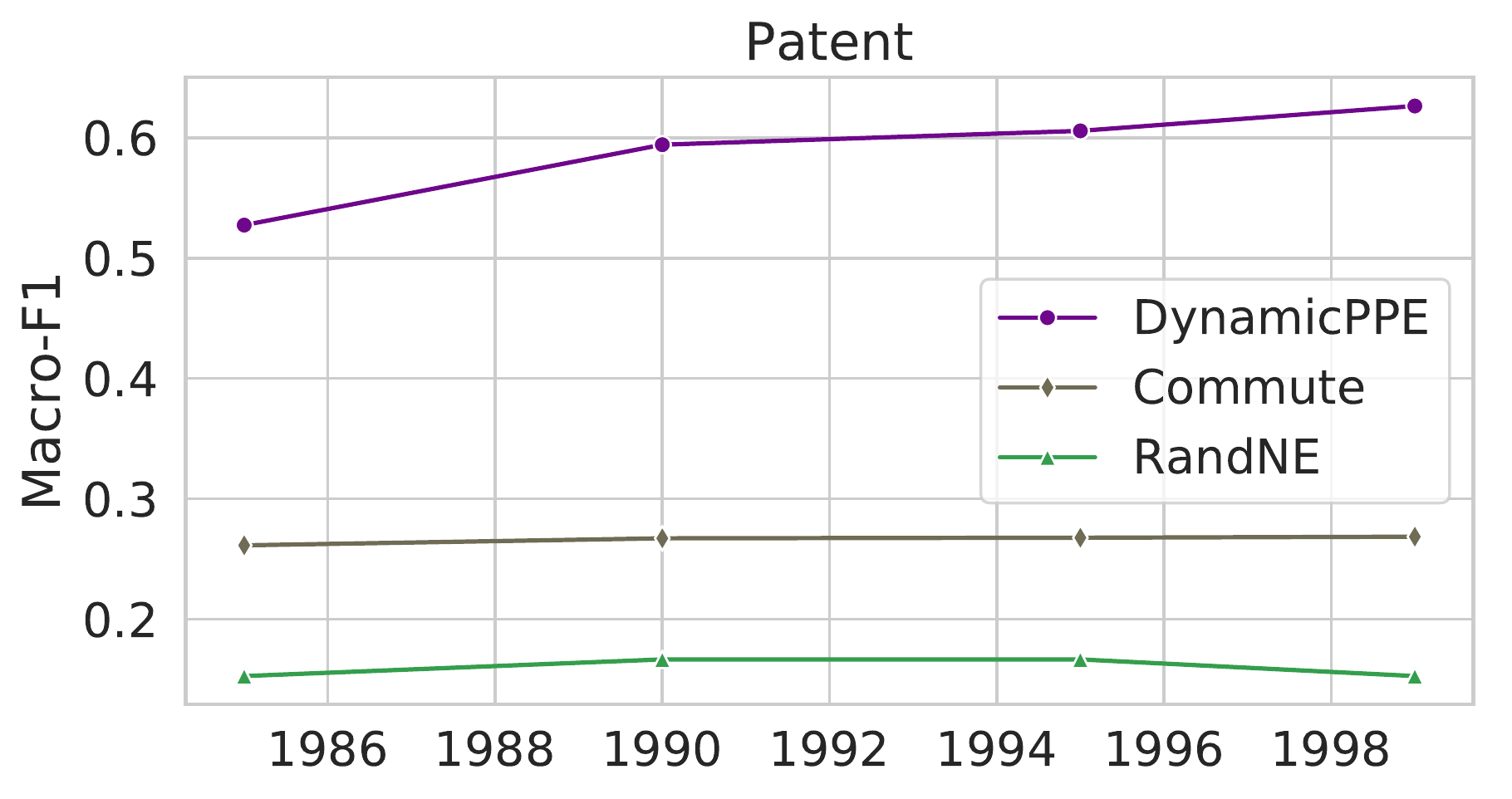}%
	\end{subfigure}
	\begin{subfigure}{.33\textwidth}
		\includegraphics[width=5cm]{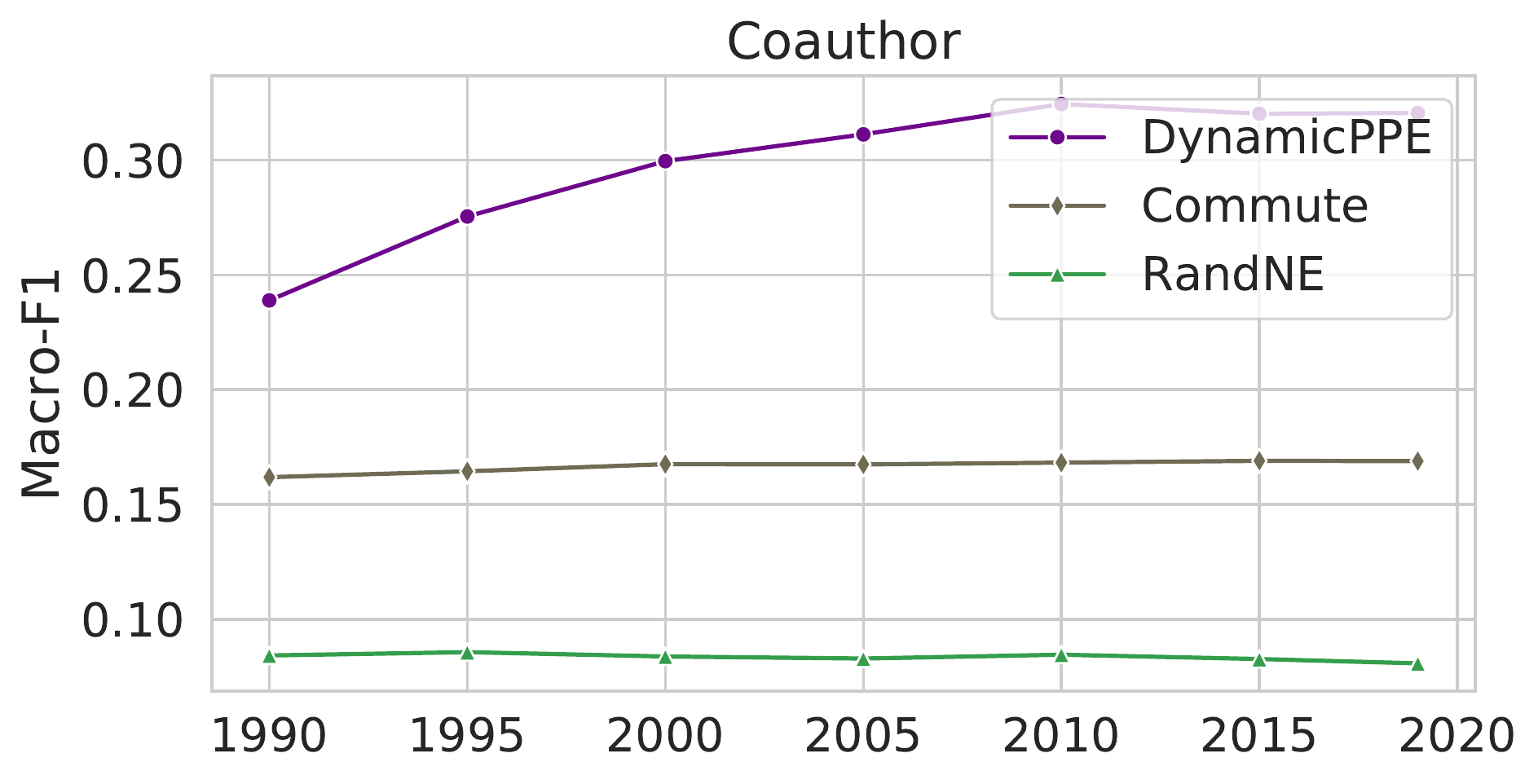}%
	\end{subfigure}

    \caption{Macro-F1 scores of node classification as a function of time. The results of the small and large graphs are on the first and second row respectively (dim=512). Our proposed methods achieves the best performance on the last snapshot when all edges arrived and performance curve matches to the static methods as the graph becomes more and more complete.}
    \label{fig:nc-full-macro-f1-over-time}
\end{figure*}

Table \ref{tab:nc-small-table} shows the classification results on the final snapshot. When we restrict the dimension to be 128, we observe that static methods outperform all dynamic methods. However, the static methods independently model each snapshot and the running time grows with number of snapshots, regardless of the number of edge changes.
In addition, our proposed method (DynPPE.) outperforms other dynamic baselines, except that in the \textit{academic} graph, where DynTri. is slightly better. However, their CPU time is 13 times more than ours as shown in Table \ref{tab:runtime-small}. 
According to the Johnson-Lindenstrauss lemma\cite{johnson1984extensions,dasgupta1999elementary}, we suspect that the poor result of RandNE is partially caused by the dimension size. As we increase the dimension to 512, we see a great performance improvement of RandNE. Also, our proposed method takes the same advantage, and outperform all others baselines in F1-score. Specifically, the increase of dimension makes hash projection step (Line 12-13 in Algorithm \ref{algo:dynamic-sne}) retain more information and yield better embeddings. We attach the visualization of embeddings in Appendix.\ref{tsne-embedding}.

The first row in the Figure \ref{fig:nc-full-macro-f1-over-time} shows the F1-scores in each snapshot when the dimension is 512. We observe that in the earlier snapshots where many edges are not arrived, the performance of DynamicTriad \cite{zhou2018dynamic} is better. One possible reason is that it models the dynamics of the graph, thus the node feature is more robust to the "missing" of the future links. While other methods, including ours, focus on an online feature updates incurred by the edge changes without explicitly modeling the temporal dynamics of a graph. Meanwhile, the performance curve of our method matches to the static methods, demonstrating the embedding quality of the intermediate snapshots is comparable to the state-of-the-art.

Table \ref{tab:runtime-small} shows the CPU time of each method (Dim=512). As we expected, static methods is very expensive as they calculate embeddings for each snapshot individually. Although our method is not blazingly fast compared to RandNE, it has a good trade-off between running time and embedding quality, especially without much hyper-parameter tuning. Most importantly, it can be easily parallelized by distributing the calculation of each node to clusters.

We also conduct experiment on large scale graphs (\textit{EnWiki20, Patent, Coauthor}). We keep track of the vertices in a subset containing $|S|=3,000$ nodes randomly selected from the first snapshot in each dataset, and similarly evaluate the quality of the embeddings of each snapshot on the node classification task. Due to scale of the graph, we compare our method against \textbf{RandNE} \cite{zhang2018billion} and an fast heuristic \textbf{Algorithm \ref{algo:commute}}. 
 Our method can calculate the embeddings for a subset of useful nodes only, while other methods have to calculate the embeddings of all nodes, which is not necessary under our scenario detailed in Sec. \ref{sec:case-study}. The second row in Figure \ref{fig:nc-full-macro-f1-over-time} shows that our proposed method has the best performance. 

\begin{table}[!h]
\caption{Total CPU time for large graphs (in second) }
\centering
\begin{tabular}{p{0.15\textwidth}|p{0.08\textwidth}|p{0.08\textwidth}|p{0.08\textwidth}}\toprule 
  & Enwiki20 & Patent & Coauthor \\
\hline 
 Commute & 6702.1 & 639.94 & 1340.74 \\
\hline 
 RandNE-II & 47992.81 & 6524.04 & 20771.19 \\
\hline 
 DynPPE. & 1538215.88 & 139222.01 & 411708.9 \\
\hline 
 DynPPE (Per-node) & 512.73 & 46.407 & 137.236 \\
 \bottomrule
\end{tabular}
\label{tab:runtime-large}
\end{table}

Table \ref{tab:runtime-large} shows the total CPU time of each method (Dim=512). Although total CPU time of our proposed method seems to be the greatest, the average CPU time for one node (as shown in row 4) is significantly smaller. This benefits a lot of downstream applications where only a subset of nodes are interesting to people in a large dynamic graph. For example, if we want to monitor the weekly embedding changes of a single node (e.g., the city of Wuhan, China) in English Wikipedia network from year 2020 to 2021, we can have the results in roughly 8.5 minutes. 
Meanwhile, other baselines have to calculate the embeddings of all nodes, and this expensive calculation may become the bottleneck in a downstream application with real-time constraints. To demonstrate the usefulness of our method, we conduct a case study in the following subsection.

\subsection{Change Detection}
\label{sec:case-study}

Thanks to the contributors timely maintaining Wikipedia, we believe that the evolution of the Wikipedia graph reflects how the real world is going on. We assume that when the structure of a node greatly changes, there must be underlying interesting events or anomalies. Since the structural changes can be well-captured by graph embedding, we use our proposed method to investigate whether anomalies happened to the Chinese cities during the COVID-19 pandemic (from Jan. 2020 to Dec. 2020).

\begin{table*}[!ht]
\caption{The top cities ranked by the z-score along time. The corresponding news titles are from the news in each time period. $d(v)$ is node degree, $\Delta d(v)$ is the degree changes from the previous timestamp.}
\centering
\small
\begin{tabular}{p{0.05\textwidth}|p{0.082\textwidth}|p{0.03\textwidth}|p{0.035\textwidth}|p{0.05\textwidth}|p{0.68\textwidth}}
\toprule 
 Date & City & $d(v)$ & $\Delta d(v)$ & Z-score & Top News Title \\
\hline 
 1/22/20 & Wuhan & 2890 & 54 & 2.210 & NBC: "New virus prompts U.S. to screen passengers from Wuhan, China" \\
\hline 
 2/2/20 & Wuhan & 2937 & 47 & 1.928 & WSJ: "U.S. Sets Evacuation Plan From Coronavirus-Infected Wuhan" \\
\hline 
 2/13/20 & Hohhot & 631 & 20 & 1.370 & Poeple.cn: "26 people in Hohhot were notified of dereliction of duty for prevention and control, and the director of the Health Commission was removed" (Translated from Chinese). \\
\hline 
 2/24/20 & Wuhan & 3012 & 38 & 2.063 & USA Today: "Coronavirus 20 times more lethal than the flu? Death toll passes 2,000" \\
\hline 
 3/6/20 & Wuhan & 3095 & 83 &  1.723 & Reuters: "Infelctions may drop to zero by end-March in Wuhan: Chinese government expert" \\
\hline 
 3/17/20 & Wuhan & 3173 & 78 & 1.690 & NYT: "Politicians Use of 'Wuhan Virus' Starts a Debate Health Experets Wanted to Avoid" \\
\hline 
 3/28/20 & Zhangjiakou & 517  & 15 & 1.217 & "Logo revealed for Freestyle Ski and Snowboard World Championships in Zhangjiakou" \\
\hline 
 4/8/20 & Wuhan & 3314  & 47               & 2.118 & CNN:"China lifts 76-day lockdown on Wuhan as city reemerges from conronavirus crisis" \\
\hline 
 4/19/20 & Luohe          & 106 & 15 & 2.640 & Forbes: "The Chinese Billionaire Whose Company Owns Troubled Pork Processor Smithfield Foods"  \\
\hline 
 4/30/20 & Zunhua & 52 & 17 & 2.449 & XINHUA: "Export companies resume production in Zunhua, Hebei" \\
\hline 
 5/11/20  & Shulan & 88 & 46      & 2.449 & CGTN: "NE China's Shulan City to reimpose community lockdown in 'wartime' battle against COVID-19" \\
 \bottomrule
\end{tabular}
\label{tab:z-score-city-changes}
\end{table*}

\noindent \textbf{Changes of major Chinese Cities\quad} We target nine Chinese major cities (\textit{Shanghai, Guangzhou, Nanjing, Beijing, Tianjin, \textbf{Wuhan}, Shenzhen, Chengdu, Chongqing)} and keep track of the embeddings in a 10-day time window.  From our prior knowledge, we expect that \textit{Wuhan} should greatly change since it is the first reported place of the COVID-19 outbreak in early Jan. 2020.

\begin{figure}[!h]
	\centering
	\begin{subfigure}{.23\textwidth}
		\includegraphics[width=4.3cm]{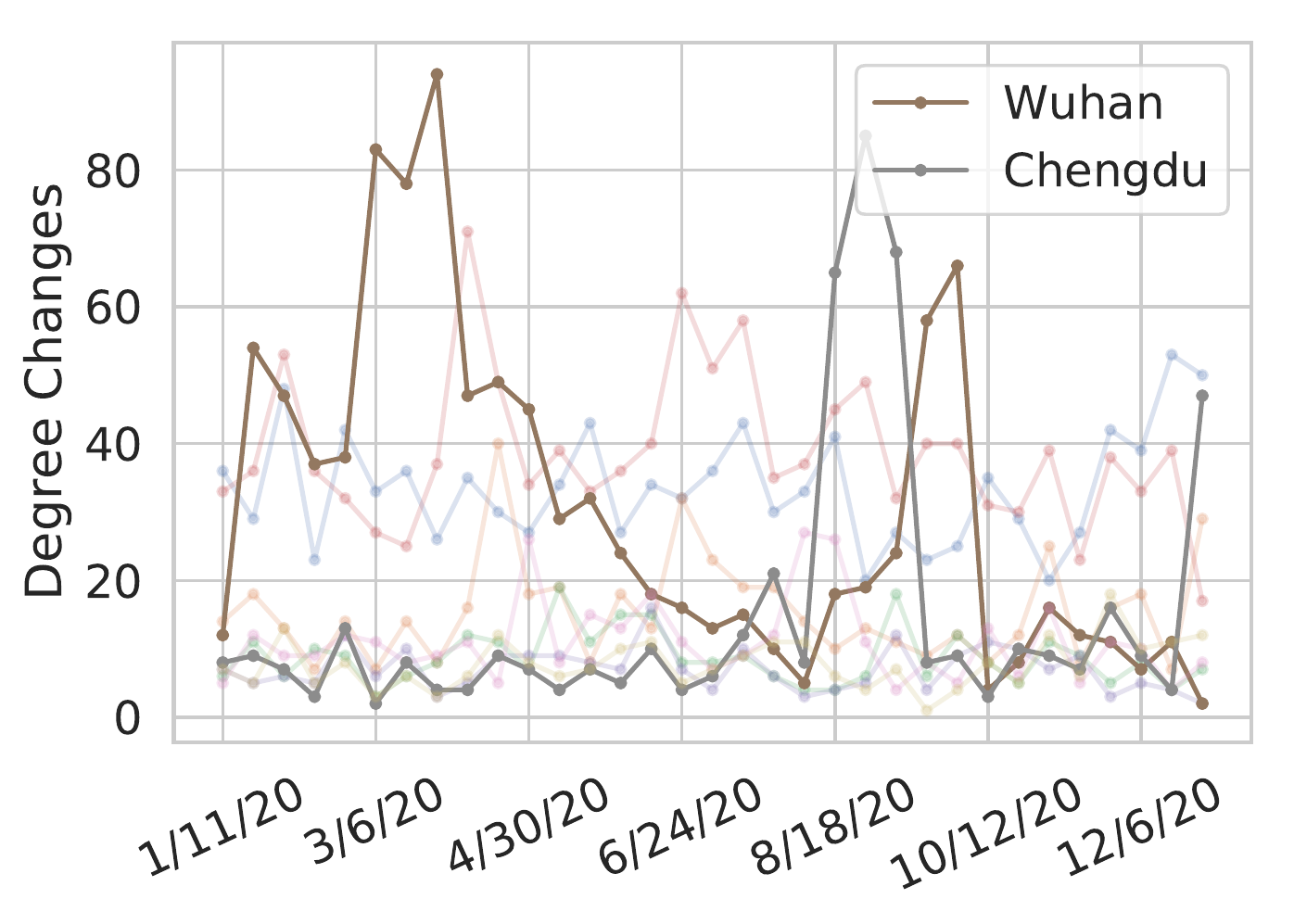}%
	\end{subfigure}
	\begin{subfigure}{.23\textwidth}
		\includegraphics[width=4.8cm]{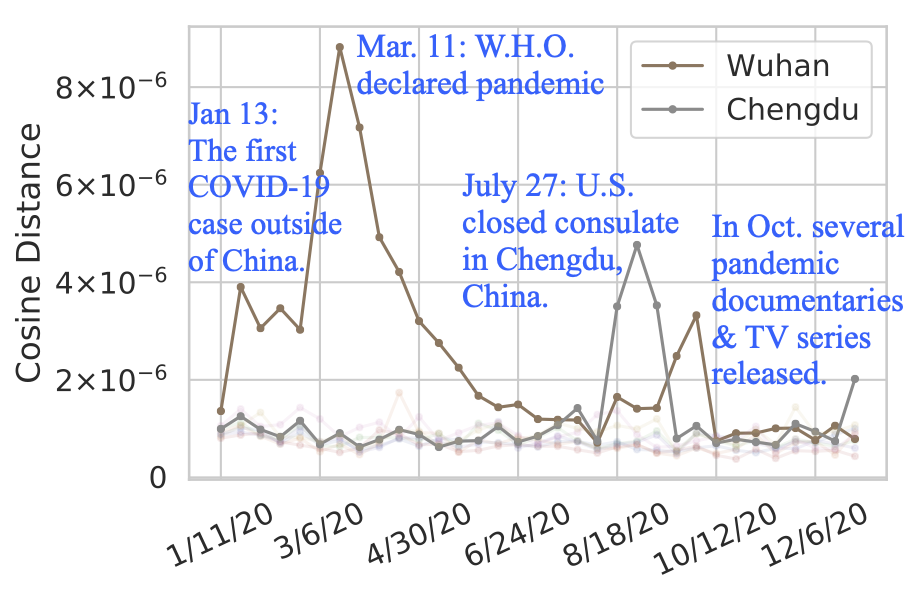}%
	\end{subfigure}
    \caption{The changes of major Chinese cities in 2020. \textit{Left}:Changes in Node Degree. \textit{Right}: Changes in Cosine distance }
    \label{fig:cs-changes}
\end{figure}

Figure.\ref{fig:cs-changes}(a) shows the node degree changes of each city every 10 days. The curve is quite noisy, but we can still see several major peaks from \textit{Wuhan} around 3/6/20 and \textit{Chengdu} around 8/18/20. When using embedding changes \footnote{Again, the embedding movement $\operatorname{Dist}(\cdot,\cdot)$ is defined as $1-\cos(\cdot,\cdot)$}  as the measurement,  Figure.\ref{fig:cs-changes} (b) provides a more clear view of the changed cities. We observe strong signals from the curve of \textit{Wuhan}, correlating to the initial COVID-19 outbreak and the declaration of pandemic \footnote{COIVD-19: \url{https://www.who.int/news/item/27-04-2020-who-timeline---covid-19}}. 
In addition, we observed an peak from the curve of \textit{Chengdu} around 8/18/20 when U.S. closed the consulate in Chengdu, China, reflecting the U.S.-China diplomatic tension\footnote{US Consulate:  \url{https://china.usembassy-china.org.cn/embassy-consulates/chengdu/}}.

\noindent \textbf{Top changed city along time\quad} We keep track of the embedding movement of 533 Chinese cities from year 2020 to 2021 in a 10-day time window, and filter out the inactive records by setting the threshold of degree changes (e.g. greater than 10). The final results are 51 Chinese cities from the major ones listed above and less famous cities like \textit{Zhangjiakou, Hohhot, Shulan, ...}.

Furthermore, we define the z-score of a target node $u$ as $Z_t(u)$ based on the embedding changes within a specific period of time from $t-1$ to $t$.
\begin{align*}
 Z_t(u|u\in S) &= \frac{Dist(w_u^{t}, w_u^{t-1}) - \mu  }{\sigma} \\
    \mu = \frac{1}{|S|} \sum_{u' \in S} Dist(w_{u'}^{t}, w_{u'}^{t-1})&,\  \ 
    \sigma = \sqrt{\frac{1}{|S|} \sum_{u' \in S} (Dist(w_{u'}^{t} w_{u'}^{t-1}) - \mu)^2 } \\
\end{align*}

In Table \ref{tab:z-score-city-changes}, we list the highest ranked cities by the z-score from Jan.22 to May 11, 2020. In addition, we also attach the top news titles corresponding to the city within each specific time period. We found that Wuhan generally appears more frequently as the situation of the pandemic kept changing. Meanwhile, we found the appearance of \textit{Hohhot} and \textit{Shulan} reflects the time when COVID-19 outbreak happened in those cities. We also discovered cities unrelated to the pandemic. For example, \textit{Luohe}, on 4/19/20, turns out to be the city where the headquarter of the organization which acquired Smithfield Foods (as mentioned in the news). In addition, \textit{Zhangjiakou}, on 3/28/20, is the city, which will host World Snowboard competition, released the Logo of that competition.

\section{Discussion and Conclusion}
\label{sec:conclusion}

In this paper, we propose a new method to learn dynamic node embeddings over large-scale dynamic networks for a subset of interesting nodes. Our proposed method has time complexity that is linear-dependent on the number of edges $m$ but independent on the number of nodes $n$. This makes our method applicable to applications of subset representations on very large-scale dynamic graphs. For the future work, as shown in \citet{trivedi2019dyrep}, there are two dynamics on dynamic graph data, \textit{structure evolution} and dynamics of \textit{node interaction}.  It would be interesting to study how one can incorporate dynamic of node interaction into our model.  It is also worth to study how different version of local push operation affect the performance of our method.

\begin{acks}
This work was partially supported by NSF grants IIS-1926781, IIS-1927227, IIS-1546113 and OAC-1919752.
\end{acks}

\bibliographystyle{ACM-Reference-Format}
\bibliography{references.bib}


\begin{thebibliography}{50}


\ifx \showCODEN    \undefined \def \showCODEN     #1{\unskip}     \fi
\ifx \showDOI      \undefined \def \showDOI       #1{#1}\fi
\ifx \showISBNx    \undefined \def \showISBNx     #1{\unskip}     \fi
\ifx \showISBNxiii \undefined \def \showISBNxiii  #1{\unskip}     \fi
\ifx \showISSN     \undefined \def \showISSN      #1{\unskip}     \fi
\ifx \showLCCN     \undefined \def \showLCCN      #1{\unskip}     \fi
\ifx \shownote     \undefined \def \shownote      #1{#1}          \fi
\ifx \showarticletitle \undefined \def \showarticletitle #1{#1}   \fi
\ifx \showURL      \undefined \def \showURL       {\relax}        \fi
\providecommand\bibfield[2]{#2}
\providecommand\bibinfo[2]{#2}
\providecommand\natexlab[1]{#1}
\providecommand\showeprint[2][]{arXiv:#2}

\bibitem[\protect\citeauthoryear{Andersen, Borgs, Chayes, Hopcraft, Mirrokni,
  and Teng}{Andersen et~al\mbox{.}}{2007a}]%
        {andersen2007local}
\bibfield{author}{\bibinfo{person}{Reid Andersen}, \bibinfo{person}{Christian
  Borgs}, \bibinfo{person}{Jennifer Chayes}, \bibinfo{person}{John Hopcraft},
  \bibinfo{person}{Vahab~S Mirrokni}, {and} \bibinfo{person}{Shang-Hua Teng}.}
  \bibinfo{year}{2007}\natexlab{a}.
\newblock \showarticletitle{Local computation of PageRank contributions}. In
  \bibinfo{booktitle}{\emph{International Workshop on Algorithms and Models for
  the Web-Graph}}. Springer, \bibinfo{pages}{150--165}.
\newblock


\bibitem[\protect\citeauthoryear{Andersen, Chung, and Lang}{Andersen
  et~al\mbox{.}}{2006}]%
        {andersen2006local}
\bibfield{author}{\bibinfo{person}{Reid Andersen}, \bibinfo{person}{Fan Chung},
  {and} \bibinfo{person}{Kevin Lang}.} \bibinfo{year}{2006}\natexlab{}.
\newblock \showarticletitle{Local graph partitioning using pagerank vectors}.
  In \bibinfo{booktitle}{\emph{2006 47th Annual IEEE Symposium on Foundations
  of Computer Science (FOCS'06)}}. IEEE, \bibinfo{pages}{475--486}.
\newblock


\bibitem[\protect\citeauthoryear{Andersen, Chung, and Lang}{Andersen
  et~al\mbox{.}}{2007b}]%
        {andersen2007using}
\bibfield{author}{\bibinfo{person}{Reid Andersen}, \bibinfo{person}{Fan Chung},
  {and} \bibinfo{person}{Kevin Lang}.} \bibinfo{year}{2007}\natexlab{b}.
\newblock \showarticletitle{Using pagerank to locally partition a graph}.
\newblock \bibinfo{journal}{\emph{Internet Mathematics}} \bibinfo{volume}{4},
  \bibinfo{number}{1} (\bibinfo{year}{2007}), \bibinfo{pages}{35--64}.
\newblock


\bibitem[\protect\citeauthoryear{Berger-Wolf and Saia}{Berger-Wolf and
  Saia}{2006}]%
        {berger2006framework}
\bibfield{author}{\bibinfo{person}{Tanya~Y Berger-Wolf} {and}
  \bibinfo{person}{Jared Saia}.} \bibinfo{year}{2006}\natexlab{}.
\newblock \showarticletitle{A framework for analysis of dynamic social
  networks}. In \bibinfo{booktitle}{\emph{Proceedings of the 12th ACM SIGKDD
  international conference on Knowledge discovery and data mining}}.
  \bibinfo{pages}{523--528}.
\newblock


\bibitem[\protect\citeauthoryear{Berkhin}{Berkhin}{2006}]%
        {berkhin2006bookmark}
\bibfield{author}{\bibinfo{person}{Pavel Berkhin}.}
  \bibinfo{year}{2006}\natexlab{}.
\newblock \showarticletitle{Bookmark-coloring algorithm for personalized
  pagerank computing}.
\newblock \bibinfo{journal}{\emph{Internet Mathematics}} \bibinfo{volume}{3},
  \bibinfo{number}{1} (\bibinfo{year}{2006}), \bibinfo{pages}{41--62}.
\newblock


\bibitem[\protect\citeauthoryear{Chen, Sultan, Tian, Chen, and Skiena}{Chen
  et~al\mbox{.}}{2019}]%
        {chen2019fast}
\bibfield{author}{\bibinfo{person}{Haochen Chen}, \bibinfo{person}{Syed~Fahad
  Sultan}, \bibinfo{person}{Yingtao Tian}, \bibinfo{person}{Muhao Chen}, {and}
  \bibinfo{person}{Steven Skiena}.} \bibinfo{year}{2019}\natexlab{}.
\newblock \showarticletitle{Fast and Accurate Network Embeddings via Very
  Sparse Random Projection}. In \bibinfo{booktitle}{\emph{Proceedings of the
  28th ACM International Conference on Information and Knowledge Management}}.
  \bibinfo{pages}{399--408}.
\newblock


\bibitem[\protect\citeauthoryear{Clement, Bierbaum, O'Keeffe, and
  Alemi}{Clement et~al\mbox{.}}{2019}]%
        {clement2019use}
\bibfield{author}{\bibinfo{person}{Colin~B Clement}, \bibinfo{person}{Matthew
  Bierbaum}, \bibinfo{person}{Kevin~P O'Keeffe}, {and}
  \bibinfo{person}{Alexander~A Alemi}.} \bibinfo{year}{2019}\natexlab{}.
\newblock \showarticletitle{On the Use of ArXiv as a Dataset}.
\newblock \bibinfo{journal}{\emph{arXiv preprint arXiv:1905.00075}}
  (\bibinfo{year}{2019}).
\newblock


\bibitem[\protect\citeauthoryear{Consonni, Laniado, and Montresor}{Consonni
  et~al\mbox{.}}{2019}]%
        {consonni2019wikilinkgraphs}
\bibfield{author}{\bibinfo{person}{Cristian Consonni}, \bibinfo{person}{David
  Laniado}, {and} \bibinfo{person}{Alberto Montresor}.}
  \bibinfo{year}{2019}\natexlab{}.
\newblock \showarticletitle{WikiLinkGraphs: A complete, longitudinal and
  multi-language dataset of the Wikipedia link networks}. In
  \bibinfo{booktitle}{\emph{Proceedings of the International AAAI Conference on
  Web and Social Media}}, Vol.~\bibinfo{volume}{13}. \bibinfo{pages}{598--607}.
\newblock


\bibitem[\protect\citeauthoryear{Dasgupta and Gupta}{Dasgupta and
  Gupta}{1999}]%
        {dasgupta1999elementary}
\bibfield{author}{\bibinfo{person}{Sanjoy Dasgupta} {and}
  \bibinfo{person}{Anupam Gupta}.} \bibinfo{year}{1999}\natexlab{}.
\newblock \showarticletitle{An elementary proof of the Johnson-Lindenstrauss
  lemma}.
\newblock \bibinfo{journal}{\emph{International Computer Science Institute,
  Technical Report}} \bibinfo{volume}{22}, \bibinfo{number}{1}
  (\bibinfo{year}{1999}), \bibinfo{pages}{1--5}.
\newblock


\bibitem[\protect\citeauthoryear{Du, Wang, Song, Lu, and Wang}{Du
  et~al\mbox{.}}{2018}]%
        {du2018dynamic}
\bibfield{author}{\bibinfo{person}{Lun Du}, \bibinfo{person}{Yun Wang},
  \bibinfo{person}{Guojie Song}, \bibinfo{person}{Zhicong Lu}, {and}
  \bibinfo{person}{Junshan Wang}.} \bibinfo{year}{2018}\natexlab{}.
\newblock \showarticletitle{Dynamic Network Embedding: An Extended Approach for
  Skip-gram based Network Embedding.}. In \bibinfo{booktitle}{\emph{IJCAI}},
  Vol.~\bibinfo{volume}{2018}. \bibinfo{pages}{2086--2092}.
\newblock


\bibitem[\protect\citeauthoryear{Gleich}{Gleich}{2015}]%
        {gleich2015pagerank}
\bibfield{author}{\bibinfo{person}{David~F Gleich}.}
  \bibinfo{year}{2015}\natexlab{}.
\newblock \showarticletitle{PageRank beyond the Web}.
\newblock \bibinfo{journal}{\emph{Siam Review}} \bibinfo{volume}{57},
  \bibinfo{number}{3} (\bibinfo{year}{2015}), \bibinfo{pages}{321--363}.
\newblock


\bibitem[\protect\citeauthoryear{Grover and Leskovec}{Grover and
  Leskovec}{2016}]%
        {grover2016node2vec}
\bibfield{author}{\bibinfo{person}{Aditya Grover} {and} \bibinfo{person}{Jure
  Leskovec}.} \bibinfo{year}{2016}\natexlab{}.
\newblock \showarticletitle{node2vec: Scalable feature learning for networks}.
  In \bibinfo{booktitle}{\emph{Proceedings of the 22nd ACM SIGKDD international
  conference on Knowledge discovery and data mining}}.
  \bibinfo{pages}{855--864}.
\newblock


\bibitem[\protect\citeauthoryear{Guo, Li, Sha, and Tan}{Guo
  et~al\mbox{.}}{2017}]%
        {guo2017parallel}
\bibfield{author}{\bibinfo{person}{Wentian Guo}, \bibinfo{person}{Yuchen Li},
  \bibinfo{person}{Mo Sha}, {and} \bibinfo{person}{Kian-Lee Tan}.}
  \bibinfo{year}{2017}\natexlab{}.
\newblock \showarticletitle{Parallel personalized pagerank on dynamic graphs}.
\newblock \bibinfo{journal}{\emph{Proceedings of the VLDB Endowment}}
  \bibinfo{volume}{11}, \bibinfo{number}{1} (\bibinfo{year}{2017}),
  \bibinfo{pages}{93--106}.
\newblock


\bibitem[\protect\citeauthoryear{Hall, Jaffe, and Trajtenberg}{Hall
  et~al\mbox{.}}{2001}]%
        {hall2001nber}
\bibfield{author}{\bibinfo{person}{Bronwyn~H Hall}, \bibinfo{person}{Adam~B
  Jaffe}, {and} \bibinfo{person}{Manuel Trajtenberg}.}
  \bibinfo{year}{2001}\natexlab{}.
\newblock \bibinfo{booktitle}{\emph{The NBER patent citation data file:
  Lessons, insights and methodological tools}}.
\newblock \bibinfo{type}{{T}echnical {R}eport}. \bibinfo{institution}{National
  Bureau of Economic Research}.
\newblock


\bibitem[\protect\citeauthoryear{Hamilton, Ying, and Leskovec}{Hamilton
  et~al\mbox{.}}{2017a}]%
        {hamilton2017inductive}
\bibfield{author}{\bibinfo{person}{Will Hamilton}, \bibinfo{person}{Zhitao
  Ying}, {and} \bibinfo{person}{Jure Leskovec}.}
  \bibinfo{year}{2017}\natexlab{a}.
\newblock \showarticletitle{Inductive representation learning on large graphs}.
  In \bibinfo{booktitle}{\emph{Advances in neural information processing
  systems}}. \bibinfo{pages}{1024--1034}.
\newblock


\bibitem[\protect\citeauthoryear{Hamilton, Ying, and Leskovec}{Hamilton
  et~al\mbox{.}}{2017b}]%
        {hamilton2017representation}
\bibfield{author}{\bibinfo{person}{William~L Hamilton}, \bibinfo{person}{Rex
  Ying}, {and} \bibinfo{person}{Jure Leskovec}.}
  \bibinfo{year}{2017}\natexlab{b}.
\newblock \showarticletitle{Representation learning on graphs: Methods and
  applications}.
\newblock \bibinfo{journal}{\emph{arXiv preprint arXiv:1709.05584}}
  (\bibinfo{year}{2017}).
\newblock


\bibitem[\protect\citeauthoryear{Hoff}{Hoff}{2007}]%
        {hoff2007modeling}
\bibfield{author}{\bibinfo{person}{Peter Hoff}.}
  \bibinfo{year}{2007}\natexlab{}.
\newblock \showarticletitle{Modeling homophily and stochastic equivalence in
  symmetric relational data}.
\newblock \bibinfo{journal}{\emph{Advances in neural information processing
  systems}}  \bibinfo{volume}{20} (\bibinfo{year}{2007}),
  \bibinfo{pages}{657--664}.
\newblock


\bibitem[\protect\citeauthoryear{Hoff, Raftery, and Handcock}{Hoff
  et~al\mbox{.}}{2002}]%
        {hoff2002latent}
\bibfield{author}{\bibinfo{person}{Peter~D Hoff}, \bibinfo{person}{Adrian~E
  Raftery}, {and} \bibinfo{person}{Mark~S Handcock}.}
  \bibinfo{year}{2002}\natexlab{}.
\newblock \showarticletitle{Latent space approaches to social network
  analysis}.
\newblock \bibinfo{journal}{\emph{Journal of the american Statistical
  association}} \bibinfo{volume}{97}, \bibinfo{number}{460}
  (\bibinfo{year}{2002}), \bibinfo{pages}{1090--1098}.
\newblock


\bibitem[\protect\citeauthoryear{Hong, Coutinho, Dey, Barab{\'a}si,
  Vogelsberger, Hernquist, and Gebhardt}{Hong et~al\mbox{.}}{2016}]%
        {hong2016discriminating}
\bibfield{author}{\bibinfo{person}{Sungryong Hong}, \bibinfo{person}{Bruno~C
  Coutinho}, \bibinfo{person}{Arjun Dey}, \bibinfo{person}{Albert-L
  Barab{\'a}si}, \bibinfo{person}{Mark Vogelsberger}, \bibinfo{person}{Lars
  Hernquist}, {and} \bibinfo{person}{Karl Gebhardt}.}
  \bibinfo{year}{2016}\natexlab{}.
\newblock \showarticletitle{Discriminating topology in galaxy distributions
  using network analysis}.
\newblock \bibinfo{journal}{\emph{Monthly Notices of the Royal Astronomical
  Society}} \bibinfo{volume}{459}, \bibinfo{number}{3} (\bibinfo{year}{2016}),
  \bibinfo{pages}{2690--2700}.
\newblock


\bibitem[\protect\citeauthoryear{Ji, He, Xu, Liu, and Zhao}{Ji
  et~al\mbox{.}}{2015}]%
        {ji2015knowledge}
\bibfield{author}{\bibinfo{person}{Guoliang Ji}, \bibinfo{person}{Shizhu He},
  \bibinfo{person}{Liheng Xu}, \bibinfo{person}{Kang Liu}, {and}
  \bibinfo{person}{Jun Zhao}.} \bibinfo{year}{2015}\natexlab{}.
\newblock \showarticletitle{Knowledge graph embedding via dynamic mapping
  matrix}. In \bibinfo{booktitle}{\emph{Proceedings of the 53rd annual meeting
  of the association for computational linguistics and the 7th international
  joint conference on natural language processing (volume 1: Long papers)}}.
  \bibinfo{pages}{687--696}.
\newblock


\bibitem[\protect\citeauthoryear{Johnson and Lindenstrauss}{Johnson and
  Lindenstrauss}{1984}]%
        {johnson1984extensions}
\bibfield{author}{\bibinfo{person}{William~B Johnson} {and}
  \bibinfo{person}{Joram Lindenstrauss}.} \bibinfo{year}{1984}\natexlab{}.
\newblock \showarticletitle{Extensions of Lipschitz mappings into a Hilbert
  space}.
\newblock \bibinfo{journal}{\emph{Contemporary mathematics}}
  \bibinfo{volume}{26}, \bibinfo{number}{189-206} (\bibinfo{year}{1984}),
  \bibinfo{pages}{1}.
\newblock


\bibitem[\protect\citeauthoryear{Kazemi and Goel}{Kazemi and Goel}{2020}]%
        {kazemi2020representation}
\bibfield{author}{\bibinfo{person}{Seyed~Mehran Kazemi} {and}
  \bibinfo{person}{Rishab Goel}.} \bibinfo{year}{2020}\natexlab{}.
\newblock \showarticletitle{Representation Learning for Dynamic Graphs: A
  Survey.}
\newblock \bibinfo{journal}{\emph{Journal of Machine Learning Research}}
  \bibinfo{volume}{21}, \bibinfo{number}{70} (\bibinfo{year}{2020}),
  \bibinfo{pages}{1--73}.
\newblock


\bibitem[\protect\citeauthoryear{Kipf and Welling}{Kipf and Welling}{2017}]%
        {kipf2017semi}
\bibfield{author}{\bibinfo{person}{Thomas~N. Kipf} {and} \bibinfo{person}{Max
  Welling}.} \bibinfo{year}{2017}\natexlab{}.
\newblock \showarticletitle{Semi-Supervised Classification with Graph
  Convolutional Networks}. In \bibinfo{booktitle}{\emph{International
  Conference on Learning Representations (ICLR)}}.
\newblock


\bibitem[\protect\citeauthoryear{Kumar, Zhang, and Leskovec}{Kumar
  et~al\mbox{.}}{2019}]%
        {kumar2019predicting}
\bibfield{author}{\bibinfo{person}{Srijan Kumar}, \bibinfo{person}{Xikun
  Zhang}, {and} \bibinfo{person}{Jure Leskovec}.}
  \bibinfo{year}{2019}\natexlab{}.
\newblock \showarticletitle{Predicting dynamic embedding trajectory in temporal
  interaction networks}. In \bibinfo{booktitle}{\emph{Proceedings of the 25th
  ACM SIGKDD International Conference on Knowledge Discovery \& Data Mining}}.
  \bibinfo{pages}{1269--1278}.
\newblock


\bibitem[\protect\citeauthoryear{Lei, Woo, Liu, and Zhang}{Lei
  et~al\mbox{.}}{2003}]%
        {lei2003schur}
\bibfield{author}{\bibinfo{person}{TG Lei}, \bibinfo{person}{CW Woo},
  \bibinfo{person}{JZ Liu}, {and} \bibinfo{person}{F Zhang}.}
  \bibinfo{year}{2003}\natexlab{}.
\newblock \showarticletitle{On the Schur complements of diagonally dominant
  matrices}.
\newblock


\bibitem[\protect\citeauthoryear{Lofgren}{Lofgren}{2015}]%
        {lofgren2015efficient}
\bibfield{author}{\bibinfo{person}{Peter Lofgren}.}
  \bibinfo{year}{2015}\natexlab{}.
\newblock \emph{\bibinfo{title}{Efficient Algorithms for Personalized
  PageRank}}.
\newblock \bibinfo{thesistype}{Ph.D. Dissertation}. \bibinfo{school}{Stanford
  University}.
\newblock


\bibitem[\protect\citeauthoryear{Lofgren, Banerjee, and Goel}{Lofgren
  et~al\mbox{.}}{2015}]%
        {lofgren2015bidirectional}
\bibfield{author}{\bibinfo{person}{Peter Lofgren}, \bibinfo{person}{Siddhartha
  Banerjee}, {and} \bibinfo{person}{Ashish Goel}.}
  \bibinfo{year}{2015}\natexlab{}.
\newblock \showarticletitle{Bidirectional PageRank Estimation: From
  Average-Case to Worst-Case}. In \bibinfo{booktitle}{\emph{Proceedings of the
  12th International Workshop on Algorithms and Models for the Web Graph-Volume
  9479}}. \bibinfo{pages}{164--176}.
\newblock


\bibitem[\protect\citeauthoryear{Nguyen, Lee, Rossi, Ahmed, Koh, and
  Kim}{Nguyen et~al\mbox{.}}{2018}]%
        {nguyen2018continuous}
\bibfield{author}{\bibinfo{person}{Giang~Hoang Nguyen},
  \bibinfo{person}{John~Boaz Lee}, \bibinfo{person}{Ryan~A Rossi},
  \bibinfo{person}{Nesreen~K Ahmed}, \bibinfo{person}{Eunyee Koh}, {and}
  \bibinfo{person}{Sungchul Kim}.} \bibinfo{year}{2018}\natexlab{}.
\newblock \showarticletitle{Continuous-time dynamic network embeddings}. In
  \bibinfo{booktitle}{\emph{Companion Proceedings of the The Web Conference
  2018}}. \bibinfo{pages}{969--976}.
\newblock


\bibitem[\protect\citeauthoryear{Perozzi, Al-Rfou, and Skiena}{Perozzi
  et~al\mbox{.}}{2014}]%
        {perozzi2014deepwalk}
\bibfield{author}{\bibinfo{person}{Bryan Perozzi}, \bibinfo{person}{Rami
  Al-Rfou}, {and} \bibinfo{person}{Steven Skiena}.}
  \bibinfo{year}{2014}\natexlab{}.
\newblock \showarticletitle{Deepwalk: Online learning of social
  representations}. In \bibinfo{booktitle}{\emph{Proceedings of the 20th ACM
  SIGKDD international conference on Knowledge discovery and data mining}}.
  \bibinfo{pages}{701--710}.
\newblock


\bibitem[\protect\citeauthoryear{Post{\u{a}}varu, Tsitsulin, de~Almeida, Tian,
  Lattanzi, and Perozzi}{Post{\u{a}}varu et~al\mbox{.}}{2020}]%
        {postuavaru2020instantembedding}
\bibfield{author}{\bibinfo{person}{{\c{S}}tefan Post{\u{a}}varu},
  \bibinfo{person}{Anton Tsitsulin}, \bibinfo{person}{Filipe
  Miguel~Gon{\c{c}}alves de Almeida}, \bibinfo{person}{Yingtao Tian},
  \bibinfo{person}{Silvio Lattanzi}, {and} \bibinfo{person}{Bryan Perozzi}.}
  \bibinfo{year}{2020}\natexlab{}.
\newblock \showarticletitle{InstantEmbedding: Efficient Local Node
  Representations}.
\newblock \bibinfo{journal}{\emph{arXiv preprint arXiv:2010.06992}}
  (\bibinfo{year}{2020}).
\newblock


\bibitem[\protect\citeauthoryear{Rossi, Gallagher, Neville, and
  Henderson}{Rossi et~al\mbox{.}}{2013}]%
        {rossi2013modeling}
\bibfield{author}{\bibinfo{person}{Ryan~A Rossi}, \bibinfo{person}{Brian
  Gallagher}, \bibinfo{person}{Jennifer Neville}, {and} \bibinfo{person}{Keith
  Henderson}.} \bibinfo{year}{2013}\natexlab{}.
\newblock \showarticletitle{Modeling dynamic behavior in large evolving
  graphs}. In \bibinfo{booktitle}{\emph{Proceedings of the sixth ACM
  international conference on Web search and data mining}}.
  \bibinfo{pages}{667--676}.
\newblock


\bibitem[\protect\citeauthoryear{Sarkar and Moore}{Sarkar and Moore}{2005a}]%
        {sarkar2005dynamic-kdd}
\bibfield{author}{\bibinfo{person}{Purnamrita Sarkar} {and}
  \bibinfo{person}{Andrew~W Moore}.} \bibinfo{year}{2005}\natexlab{a}.
\newblock \showarticletitle{Dynamic social network analysis using latent space
  models}.
\newblock \bibinfo{journal}{\emph{Acm Sigkdd Explorations Newsletter}}
  \bibinfo{volume}{7}, \bibinfo{number}{2} (\bibinfo{year}{2005}),
  \bibinfo{pages}{31--40}.
\newblock


\bibitem[\protect\citeauthoryear{Sarkar and Moore}{Sarkar and Moore}{2005b}]%
        {sarkar2005dynamic-nips}
\bibfield{author}{\bibinfo{person}{Purnamrita Sarkar} {and}
  \bibinfo{person}{Andrew~W Moore}.} \bibinfo{year}{2005}\natexlab{b}.
\newblock \showarticletitle{Dynamic social network analysis using latent space
  models}. In \bibinfo{booktitle}{\emph{Proceedings of the 18th International
  Conference on Neural Information Processing Systems}}.
  \bibinfo{pages}{1145--1152}.
\newblock


\bibitem[\protect\citeauthoryear{Sarkar, Siddiqi, and Gordon}{Sarkar
  et~al\mbox{.}}{2007}]%
        {sarkar2007latent}
\bibfield{author}{\bibinfo{person}{Purnamrita Sarkar}, \bibinfo{person}{Sajid~M
  Siddiqi}, {and} \bibinfo{person}{Geogrey~J Gordon}.}
  \bibinfo{year}{2007}\natexlab{}.
\newblock \showarticletitle{A latent space approach to dynamic embedding of
  co-occurrence data}. In \bibinfo{booktitle}{\emph{Artificial Intelligence and
  Statistics}}. \bibinfo{pages}{420--427}.
\newblock


\bibitem[\protect\citeauthoryear{Sinha, Shen, Song, Ma, Eide, Hsu, and
  Wang}{Sinha et~al\mbox{.}}{2015}]%
        {sinha2015overview}
\bibfield{author}{\bibinfo{person}{Arnab Sinha}, \bibinfo{person}{Zhihong
  Shen}, \bibinfo{person}{Yang Song}, \bibinfo{person}{Hao Ma},
  \bibinfo{person}{Darrin Eide}, \bibinfo{person}{Bo-June Hsu}, {and}
  \bibinfo{person}{Kuansan Wang}.} \bibinfo{year}{2015}\natexlab{}.
\newblock \showarticletitle{An overview of microsoft academic service (mas) and
  applications}. In \bibinfo{booktitle}{\emph{Proceedings of the 24th
  international conference on world wide web}}. \bibinfo{pages}{243--246}.
\newblock


\bibitem[\protect\citeauthoryear{Tang, Zhang, Yao, Li, Zhang, and Su}{Tang
  et~al\mbox{.}}{2008}]%
        {tang2008arnetminer}
\bibfield{author}{\bibinfo{person}{Jie Tang}, \bibinfo{person}{Jing Zhang},
  \bibinfo{person}{Limin Yao}, \bibinfo{person}{Juanzi Li}, \bibinfo{person}{Li
  Zhang}, {and} \bibinfo{person}{Zhong Su}.} \bibinfo{year}{2008}\natexlab{}.
\newblock \showarticletitle{Arnetminer: extraction and mining of academic
  social networks}. In \bibinfo{booktitle}{\emph{Proceedings of the 14th ACM
  SIGKDD international conference on Knowledge discovery and data mining}}.
  \bibinfo{pages}{990--998}.
\newblock


\bibitem[\protect\citeauthoryear{Trivedi, Farajtabar, Biswal, and Zha}{Trivedi
  et~al\mbox{.}}{2019a}]%
        {trivedi2018dyrep}
\bibfield{author}{\bibinfo{person}{Rakshit Trivedi}, \bibinfo{person}{Mehrdad
  Farajtabar}, \bibinfo{person}{Prasenjeet Biswal}, {and}
  \bibinfo{person}{Hongyuan Zha}.} \bibinfo{year}{2019}\natexlab{a}.
\newblock \showarticletitle{DyRep: Learning Representations over Dynamic
  Graphs}. In \bibinfo{booktitle}{\emph{International Conference on Learning
  Representations}}.
\newblock


\bibitem[\protect\citeauthoryear{Trivedi, Farajtabar, Biswal, and Zha}{Trivedi
  et~al\mbox{.}}{2019b}]%
        {trivedi2019dyrep}
\bibfield{author}{\bibinfo{person}{Rakshit Trivedi}, \bibinfo{person}{Mehrdad
  Farajtabar}, \bibinfo{person}{Prasenjeet Biswal}, {and}
  \bibinfo{person}{Hongyuan Zha}.} \bibinfo{year}{2019}\natexlab{b}.
\newblock \showarticletitle{DyRep: Learning Representations over Dynamic
  Graphs}. In \bibinfo{booktitle}{\emph{International Conference on Learning
  Representations}}.
\newblock


\bibitem[\protect\citeauthoryear{Tsitsulin, Mottin, Karras, and
  M{\"u}ller}{Tsitsulin et~al\mbox{.}}{2018}]%
        {tsitsulin2018verse}
\bibfield{author}{\bibinfo{person}{Anton Tsitsulin}, \bibinfo{person}{Davide
  Mottin}, \bibinfo{person}{Panagiotis Karras}, {and} \bibinfo{person}{Emmanuel
  M{\"u}ller}.} \bibinfo{year}{2018}\natexlab{}.
\newblock \showarticletitle{Verse: Versatile graph embeddings from similarity
  measures}. In \bibinfo{booktitle}{\emph{Proceedings of the 2018 World Wide
  Web Conference}}. \bibinfo{pages}{539--548}.
\newblock


\bibitem[\protect\citeauthoryear{Van~der Maaten and Hinton}{Van~der Maaten and
  Hinton}{2008}]%
        {van2008visualizing}
\bibfield{author}{\bibinfo{person}{Laurens Van~der Maaten} {and}
  \bibinfo{person}{Geoffrey Hinton}.} \bibinfo{year}{2008}\natexlab{}.
\newblock \showarticletitle{Visualizing data using t-SNE.}
\newblock \bibinfo{journal}{\emph{Journal of machine learning research}}
  \bibinfo{volume}{9}, \bibinfo{number}{11} (\bibinfo{year}{2008}).
\newblock


\bibitem[\protect\citeauthoryear{Von~Luxburg, Radl, and Hein}{Von~Luxburg
  et~al\mbox{.}}{2014}]%
        {von2014hitting}
\bibfield{author}{\bibinfo{person}{Ulrike Von~Luxburg}, \bibinfo{person}{Agnes
  Radl}, {and} \bibinfo{person}{Matthias Hein}.}
  \bibinfo{year}{2014}\natexlab{}.
\newblock \showarticletitle{Hitting and commute times in large random
  neighborhood graphs}.
\newblock \bibinfo{journal}{\emph{The Journal of Machine Learning Research}}
  \bibinfo{volume}{15}, \bibinfo{number}{1} (\bibinfo{year}{2014}),
  \bibinfo{pages}{1751--1798}.
\newblock


\bibitem[\protect\citeauthoryear{Weinberger, Dasgupta, Langford, Smola, and
  Attenberg}{Weinberger et~al\mbox{.}}{2009}]%
        {weinberger2009feature}
\bibfield{author}{\bibinfo{person}{Kilian Weinberger}, \bibinfo{person}{Anirban
  Dasgupta}, \bibinfo{person}{John Langford}, \bibinfo{person}{Alex Smola},
  {and} \bibinfo{person}{Josh Attenberg}.} \bibinfo{year}{2009}\natexlab{}.
\newblock \showarticletitle{Feature hashing for large scale multitask
  learning}. In \bibinfo{booktitle}{\emph{Proceedings of the 26th annual
  international conference on machine learning}}. \bibinfo{pages}{1113--1120}.
\newblock


\bibitem[\protect\citeauthoryear{Ying, He, Chen, Eksombatchai, Hamilton, and
  Leskovec}{Ying et~al\mbox{.}}{2018}]%
        {ying2018graph}
\bibfield{author}{\bibinfo{person}{Rex Ying}, \bibinfo{person}{Ruining He},
  \bibinfo{person}{Kaifeng Chen}, \bibinfo{person}{Pong Eksombatchai},
  \bibinfo{person}{William~L Hamilton}, {and} \bibinfo{person}{Jure Leskovec}.}
  \bibinfo{year}{2018}\natexlab{}.
\newblock \showarticletitle{Graph convolutional neural networks for web-scale
  recommender systems}. In \bibinfo{booktitle}{\emph{Proceedings of the 24th
  ACM SIGKDD International Conference on Knowledge Discovery \& Data Mining}}.
  \bibinfo{pages}{974--983}.
\newblock


\bibitem[\protect\citeauthoryear{Zang and Wang}{Zang and Wang}{2020}]%
        {zang2020neural}
\bibfield{author}{\bibinfo{person}{Chengxi Zang} {and} \bibinfo{person}{Fei
  Wang}.} \bibinfo{year}{2020}\natexlab{}.
\newblock \showarticletitle{Neural dynamics on complex networks}. In
  \bibinfo{booktitle}{\emph{Proceedings of the 26th ACM SIGKDD International
  Conference on Knowledge Discovery \& Data Mining}}.
  \bibinfo{pages}{892--902}.
\newblock


\bibitem[\protect\citeauthoryear{Zhang, Lofgren, and Goel}{Zhang
  et~al\mbox{.}}{2016a}]%
        {zhang2016approximate}
\bibfield{author}{\bibinfo{person}{Hongyang Zhang}, \bibinfo{person}{Peter
  Lofgren}, {and} \bibinfo{person}{Ashish Goel}.}
  \bibinfo{year}{2016}\natexlab{a}.
\newblock \showarticletitle{Approximate personalized pagerank on dynamic
  graphs}. In \bibinfo{booktitle}{\emph{Proceedings of the 22nd ACM SIGKDD
  International Conference on Knowledge Discovery and Data Mining}}.
  \bibinfo{pages}{1315--1324}.
\newblock


\bibitem[\protect\citeauthoryear{Zhang, Lofgren, and Goel}{Zhang
  et~al\mbox{.}}{2016b}]%
        {zhang2016approximate-extended}
\bibfield{author}{\bibinfo{person}{Hongyang Zhang}, \bibinfo{person}{Peter
  Lofgren}, {and} \bibinfo{person}{Ashish Goel}.}
  \bibinfo{year}{2016}\natexlab{b}.
\newblock \showarticletitle{Approximate Personalized PageRank on Dynamic
  Graphs}.
\newblock \bibinfo{journal}{\emph{arXiv preprint arXiv:1603.07796}}
  (\bibinfo{year}{2016}).
\newblock


\bibitem[\protect\citeauthoryear{Zhang, Cui, Li, Wang, and Zhu}{Zhang
  et~al\mbox{.}}{2018a}]%
        {zhang2018billion}
\bibfield{author}{\bibinfo{person}{Ziwei Zhang}, \bibinfo{person}{Peng Cui},
  \bibinfo{person}{Haoyang Li}, \bibinfo{person}{Xiao Wang}, {and}
  \bibinfo{person}{Wenwu Zhu}.} \bibinfo{year}{2018}\natexlab{a}.
\newblock \showarticletitle{Billion-scale network embedding with iterative
  random projection}. In \bibinfo{booktitle}{\emph{2018 IEEE International
  Conference on Data Mining (ICDM)}}. IEEE, \bibinfo{pages}{787--796}.
\newblock


\bibitem[\protect\citeauthoryear{Zhang, Cui, Pei, Wang, and Zhu}{Zhang
  et~al\mbox{.}}{2018b}]%
        {zhang2018timers}
\bibfield{author}{\bibinfo{person}{Ziwei Zhang}, \bibinfo{person}{Peng Cui},
  \bibinfo{person}{Jian Pei}, \bibinfo{person}{Xiao Wang}, {and}
  \bibinfo{person}{Wenwu Zhu}.} \bibinfo{year}{2018}\natexlab{b}.
\newblock \showarticletitle{TIMERS: Error-Bounded SVD Restart on Dynamic
  Networks}. In \bibinfo{booktitle}{\emph{Proceedings of the 32nd AAAI
  Conference on Artificial Intelligence}}. AAAI.
\newblock


\bibitem[\protect\citeauthoryear{Zhou, Yang, Ren, Wu, and Zhuang}{Zhou
  et~al\mbox{.}}{2018}]%
        {zhou2018dynamic}
\bibfield{author}{\bibinfo{person}{Lekui Zhou}, \bibinfo{person}{Yang Yang},
  \bibinfo{person}{Xiang Ren}, \bibinfo{person}{Fei Wu}, {and}
  \bibinfo{person}{Yueting Zhuang}.} \bibinfo{year}{2018}\natexlab{}.
\newblock \showarticletitle{Dynamic network embedding by modeling triadic
  closure process}. In \bibinfo{booktitle}{\emph{Proceedings of the AAAI
  Conference on Artificial Intelligence}}, Vol.~\bibinfo{volume}{32}.
\newblock


\bibitem[\protect\citeauthoryear{Zhu, Guo, Yin, Ver~Steeg, and Galstyan}{Zhu
  et~al\mbox{.}}{2016}]%
        {zhu2016scalable}
\bibfield{author}{\bibinfo{person}{Linhong Zhu}, \bibinfo{person}{Dong Guo},
  \bibinfo{person}{Junming Yin}, \bibinfo{person}{Greg Ver~Steeg}, {and}
  \bibinfo{person}{Aram Galstyan}.} \bibinfo{year}{2016}\natexlab{}.
\newblock \showarticletitle{Scalable temporal latent space inference for link
  prediction in dynamic social networks}.
\newblock \bibinfo{journal}{\emph{IEEE Transactions on Knowledge and Data
  Engineering}} \bibinfo{volume}{28}, \bibinfo{number}{10}
  (\bibinfo{year}{2016}), \bibinfo{pages}{2765--2777}.
\newblock


\end{thebibliography}

\clearpage
\appendix

\section{Proof of lemmas}

To prove Lemma \ref{lemma:1}, we first introduce the property of uniqueness of PPR $\bm \pi_s$ for any $s$.

\begin{proposition}[Uniqueness of PPR \citep{andersen2007using}]
\label{prop:1}
For any starting vector $\bm 1_s$, and any constant $\alpha \in (0,1]$, there is a unique vector $\bm \pi_s$ satisfying (\ref{equ:ppr}).
\end{proposition}
\begin{proof}
Recall the PPR equation
\begin{equation}
\bm \pi_s = \alpha \bm 1_s + (1-\alpha) \bm D^{-1}\bm A \bm \pi_s. \nonumber
\end{equation}
We can rewrite it as $(\bm I - (1-\alpha)\bm D^{-1} \bm A) \bm \pi_s = \alpha  \bm 1_s$. Notice the fact that matrix $\bm I - (1-\alpha)\bm D^{-1} \bm A$ is  strictly diagonally dominant matrix. To see this, for each $i \in \mathbb{V}$, we have $1 - (1-\alpha)\sum _{j\neq i}|1/d(i)| = \alpha > 0$. By \cite{lei2003schur}, strictly diagonally dominant matrix is always invertible.
\end{proof}

\begin{proposition}[Symmetry property \citep{lofgren2015bidirectional}]
Given any undirected graph $\mathcal{G}$, for any $\alpha \in (0,1)$ and for any node pair $(u,v)$, we have
\begin{equation}
    d(u) \pi_u(v) = d(v) \pi_v (u).
\end{equation}
\label{proposition:symmetry-property}
\end{proposition}

\begin{proof}[Proof of Lemma 7]
Assume there are $T$ iterations. For each forward push operation $t=1,2,\ldots T$, we assume the frontier node is $u_t$, the run time of one push operation is then $d(u_t)$. For total $T$ push operations, the total run time is $\sum_{i=1}^T d(u_i)$. Notice that during each push operation, the amount of $\|\bm r_s^{t-1}\|_1$ is reduced at least $\epsilon \alpha d(u_t)$, then we always have the following inequality
\begin{equation}
\epsilon \alpha d(u_t) < \| \bm r_s^{t-1} \|_1 - \| \bm r_s^t\|_1 \nonumber
\end{equation}
Apply the above inequality from $t=1,2,$ to $T$, we will have
\begin{equation}
\epsilon\alpha \sum_{t=1}^T d(u_t) \leq \|\bm r_s^0\| - \|\bm r_s^T\|_1 = 1 - \|\bm r_s \|_1,
\end{equation}
where $\bm r_s$ is the final residual vector. The total time is then $\mathcal{O}(1/\epsilon\alpha)$. To show the estimation error, we follow the idea of \cite{lofgren2015efficient}. Notice that, the forward local push algorithm always has the invariant property by Lemma \ref{lemma:1}, that is 
\begin{equation}
\pi_s(u) = p_s(u) + \sum_{v \in V} r_s(v) \pi_v(u), \forall u \in \mathbb{V}.
\end{equation}
By proposition \ref{proposition:symmetry-property}, we have
\begin{align*}
\pi_s(u) &= p_s(u) + \sum_{v \in V} r_s(v) \pi_v(u), \forall u \in \mathbb{V} \\
&= p_s(u) + \sum_{v \in V} r_s(v) \frac{d(u)}{d(v)} \pi_u(v), \forall u \in \mathbb{V} \\
&\leq p_s(u) + \sum_{v \in V} \epsilon d(u) \pi_u(v), \forall u \in \mathbb{V} = p_s(u) + \epsilon d(u),
\end{align*}
where the first inequality by the fact that $r_s(v) \leq \epsilon d(v)$ and the last equality is due to $\| \pi_u \|_1 = 1$.
\end{proof}

\begin{proposition}[\cite{zhang2016approximate}]
Let $\mathcal{G}=(V,E)$ be undirected and let $t$ be a vertex of $V$, then $\sum_{x\in V} \frac{\pi_x(t)}{d(t)} \leq 1$.
\end{proposition}
\begin{proof}
By using Proposition \ref{proposition:symmetry-property}, we have
\begin{align*}
\sum_{x\in V} \frac{\pi_x(t)}{d(t)} = \sum_{x\in V} \frac{\pi_t(x)}{d(x)} \leq \sum_{x\in V} \pi_t(x) = 1. 
\end{align*}
\end{proof}

\section{Heuristic method: Commute}

We update the embeddings by their pairwise relationship (resistance distance). The commute distance (i.e. resistance distance) $C_{u v} = H_{u v} + H_{v u}$, where  rescaled hitting time $H_{u v}$ converges to $1/d(v)$. As proved in~\cite{von2014hitting}, when the number of nodes in the graph is large enough, we can show that the commute distance tends to $1/d_v + 1/d_u$.

\begin{center}
\begin{minipage}{0.5\textwidth}
\begin{algorithm}[H]
\caption{\textsc{Commute}}
\label{algo:commute}
\begin{algorithmic}[1]
\State \textbf{Input:} An graph $\mathcal{G}^0(\mathcal{V}^0,\mathcal{E}^0)$ and embedding $\bm W^0$, dimension $d$.
\State \textbf{Output:} $\bm W^T$
\For {$e^t(u,v,t) \in \left\{e^1(u^1,v^1,t_1), \ldots, e^T(u^T,v^T,t_T)\right\}$}
\State Add $e^t$ to $\mathcal{G}^t$
\State \textbf{If } $u\notin V^{t-1}$ \textbf{ then } 
\State \qquad generate $\bm w_u^t = \mathcal{N}(\bm 0,0.1\cdot\bm I_d)$ or $ \mathcal{U}(-0.5,0.5)/d$
\State \textbf{If } $v\notin V^{t-1}$ \textbf{ then } 
\State \qquad generate $\bm w_v^t = \mathcal{N}(\bm 0,0.1\cdot\bm I_d)$ or $\mathcal{U}(-0.5,0.5)/d$
\State $\bm w_u^t = \frac{d(u)}{d(u)+1} \bm w_u^{t-1} + \frac{1}{d(u)} \bm w_v^t$
\State $\bm w_v^t = \frac{d(v)}{d(v)+1} \bm w_v^{t-1} + \frac{1}{d(v)} \bm w_u^t$
\EndFor
\State \textbf{Return} $\bm W_T$
\end{algorithmic}
\end{algorithm}
\end{minipage}
\end{center}

One can treat the Commute method, i.e. Algorithm \ref{algo:commute}, as the first-order approximation of RandNE \cite{zhang2018billion}. The embedding generated by RandNE is given as the following
\begin{equation}
    \bm U = \left( \alpha_0 \bm I + \alpha_1 \bm A + \alpha_2 \bm A^2 + \ldots + \alpha_q \bm A^q \right) \bm R,
\end{equation}
where $\bm A$ is the normalized adjacency matrix and $\bm I$ is the identity matrix. At any time $t$, the dynamic embedding of node $i$ of Commute is given by
\begin{align*}
\bm w_i^t &= \frac{d(u)}{d(u)+1} \bm w_i^{t-1} + \frac{1}{d(u)} \bm w_v^t    \\
&= \frac{1}{d(u)+1} \bm w_i^{0} + \frac{1}{|\mathcal{N}(i)|} \sum_{j\in \mathcal{N}(i)} \frac{1}{d(u)} \bm w_v^t    \\
\end{align*} 

\section{Details of data preprocessing}
\label{appendix::data}
In this section, we describe the preprocessing steps of three datasets. \textbf{Enwiki20}: In Enwiki20 graph, the edge stream is divided into 6 snapshots, containing edges before 2005, 2005-2008, ..., 2017-2020. The sampled nodes in the first snapshot fall into 5 categories. \textbf{Patent}: In full patent graph, we divide edge stream into 4 snapshots, containing patents citation links before 1985, 1985-1990,..., 1995-1999. In node classification tasks, we sampled 3,000 nodes in the first snapshot, which fall into 6 categories. 
In patent small graph, we divide into 13 snapshots with a 2-year window. All the nodes in each snapshot fall into 6 categories.%
\textbf{Coauthor graph}: In full Coauthor graph, we divide edge stream into 7 snapshots (before 1990, 1990-1995, ..., 2015-2019). The sampled nodes in the first snapshot fall into 9 categories. 
In Coauthor small graph, the edge stream is divided into 9 snapshots (before 1978, 1978-1983,..., 2013-2017). All the nodes in each snapshot have 14 labels in common.

\section{Details of parameter settings}
\label{appendix::parameter}

\textbf{Deepwalk}: number-walks=40, walk-length=40, window-size=5\\
\textbf{Node2Vec}: Same as Deepwalk, p = 0.5, q = 0.5\\
\textbf{DynamicTriad}: iteration=10, beta-smooth=10, beta-triad=10. Each input snapshot contains the previous existing edges and newly arrived edges. \\
\textbf{RandNE}: q=3, default weight for node classification $[1,1e2,1e4,1e5]$, input is the transition matrix, the output feature is normalized (l-2 norm) row-wise. \\
\textbf{DynamicPPE}: $\alpha=0.15$,  $\epsilon=0.1$, projection method=hash. our method is relatively insensitive to the hyper-parameters.

\section{Visualizations of Embeddings }
\label{tsne-embedding}
We visualize the embeddings of small scale graphs using T-SNE\cite{van2008visualizing} in Fig.\ref{fig:nc-tsne-patent-small},\ref{fig:nc-tsne-coauthor-small},\ref{fig:nc-tsne-academic-small}.

\begin{figure}[htbp]
	\centering
	\begin{subfigure}{.23\textwidth}
		\includegraphics[width=3cm]{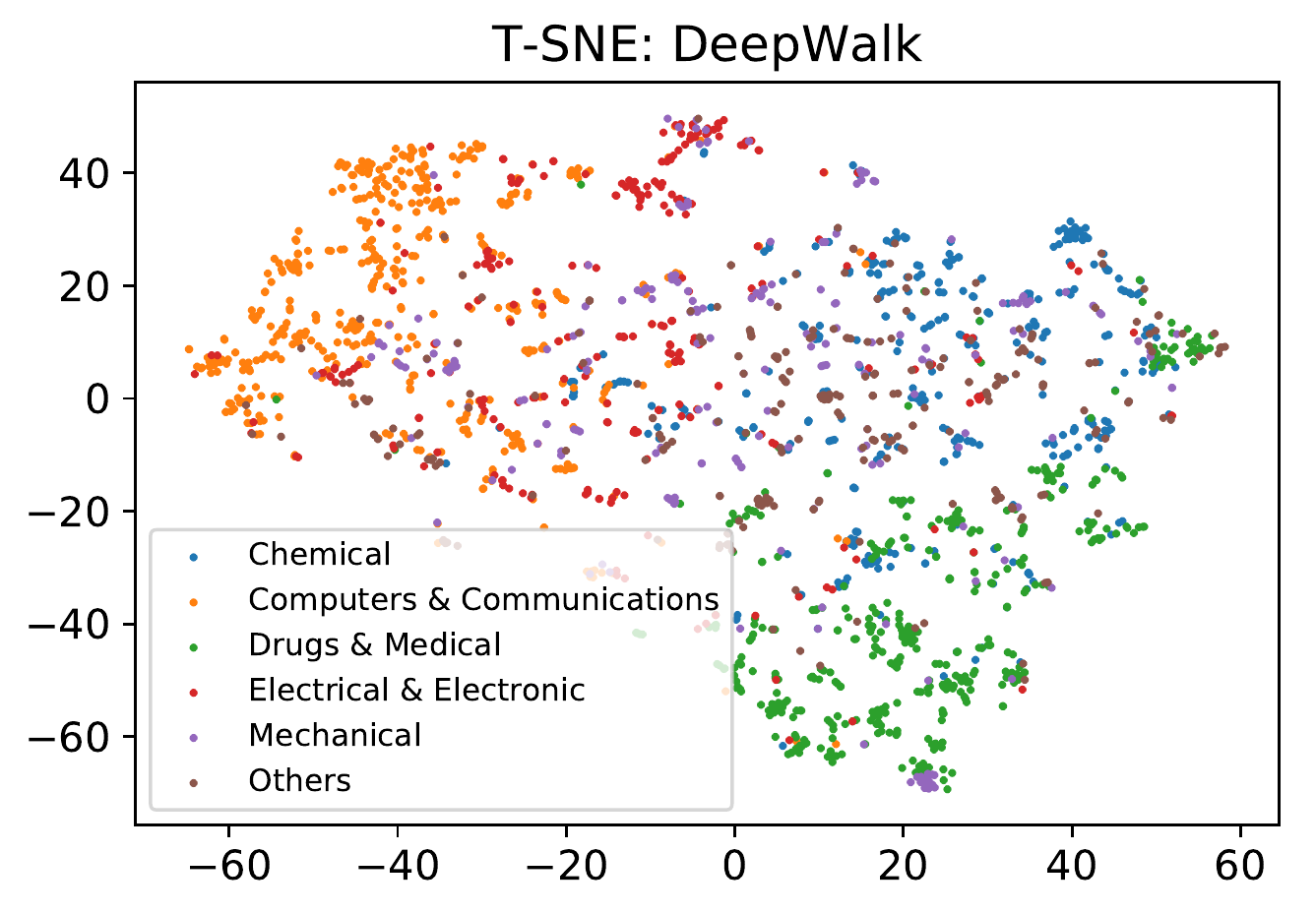}
	\end{subfigure}
	\begin{subfigure}{.23\textwidth}
		\includegraphics[width=3cm]{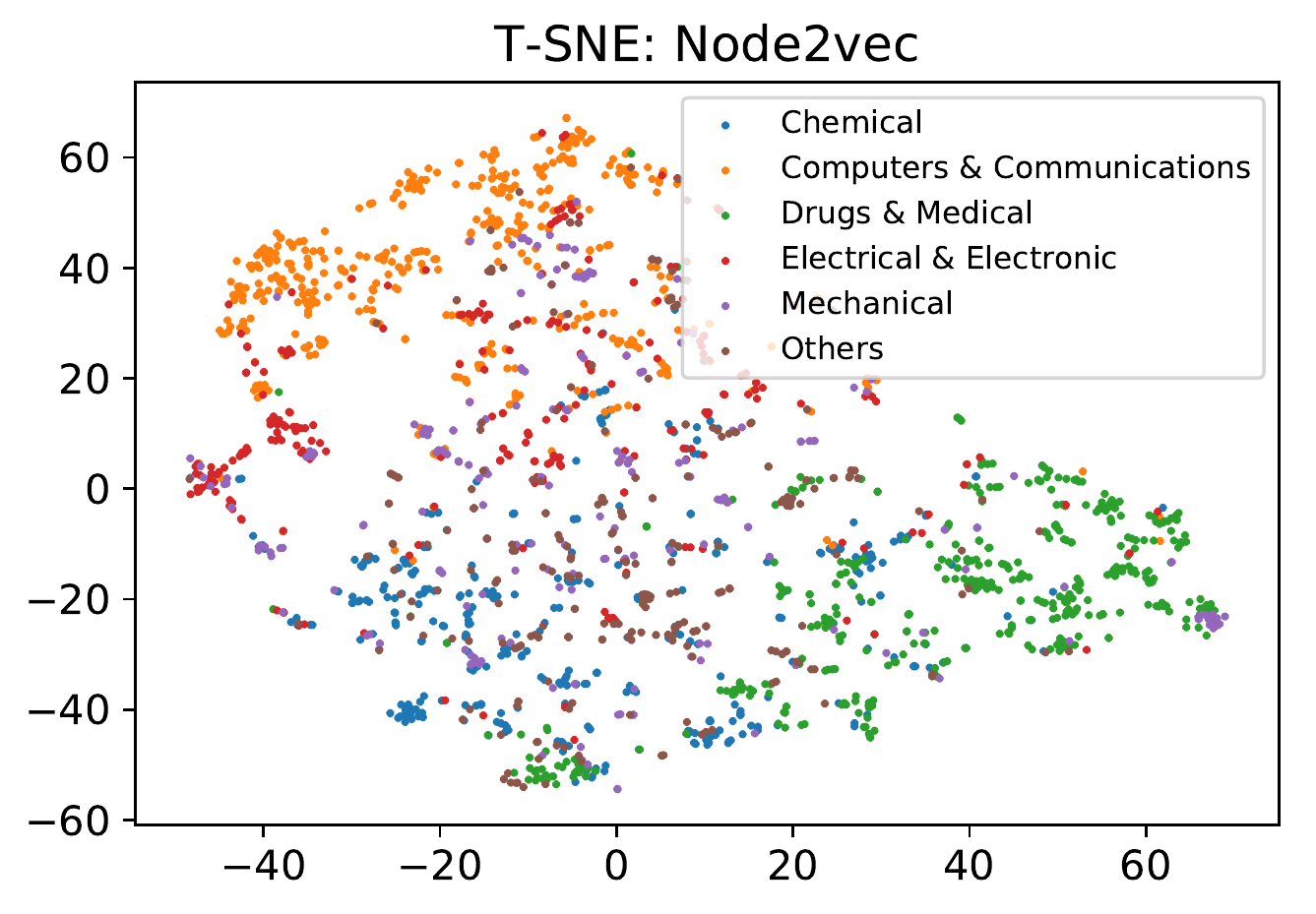}
	\end{subfigure}
	\begin{subfigure}{.23\textwidth}
		\includegraphics[width=3cm]{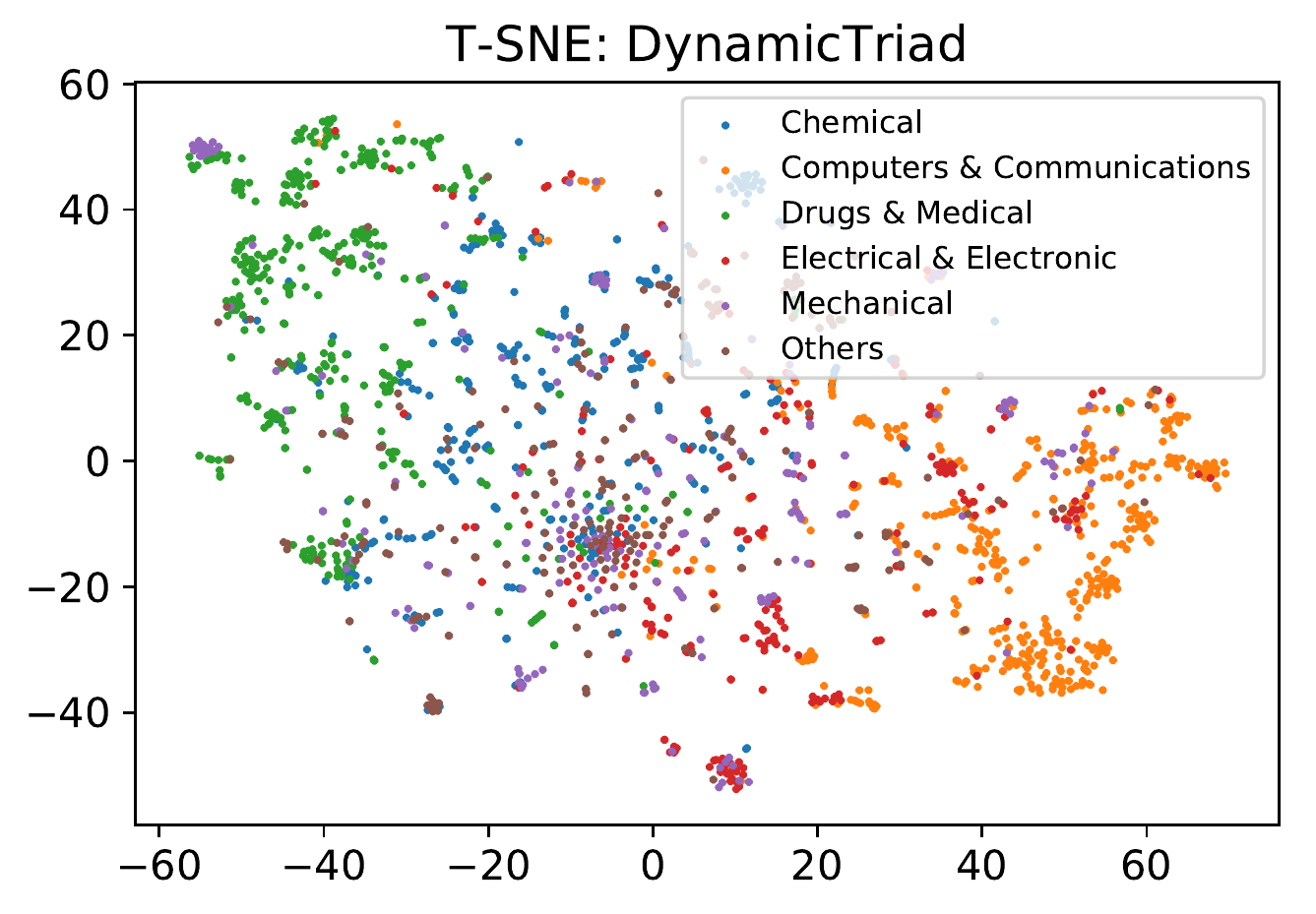}
	\end{subfigure}
	\begin{subfigure}{.23\textwidth}
		\includegraphics[width=3cm]{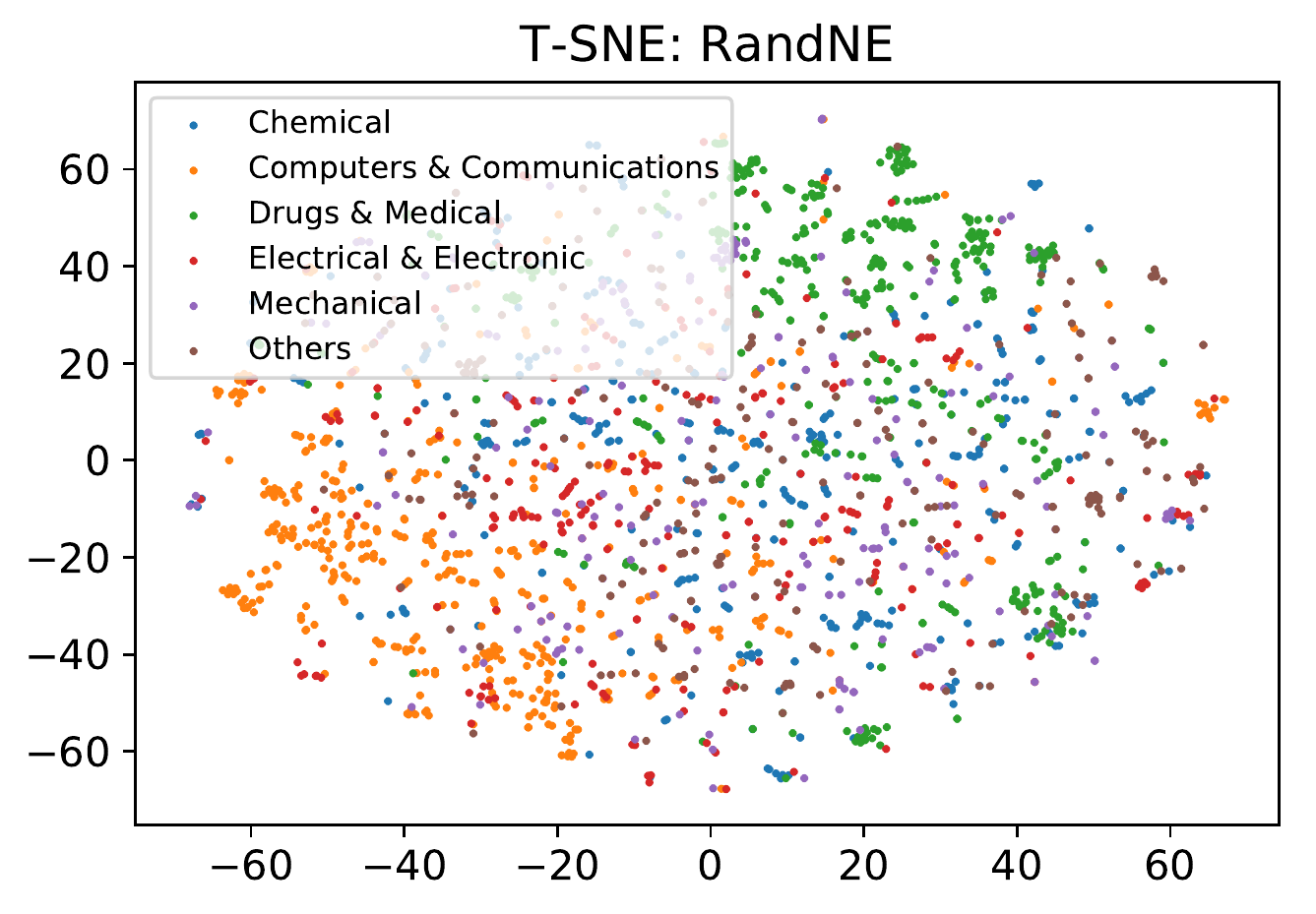}
	\end{subfigure}
	\begin{subfigure}{.23\textwidth}
		\includegraphics[width=3cm]{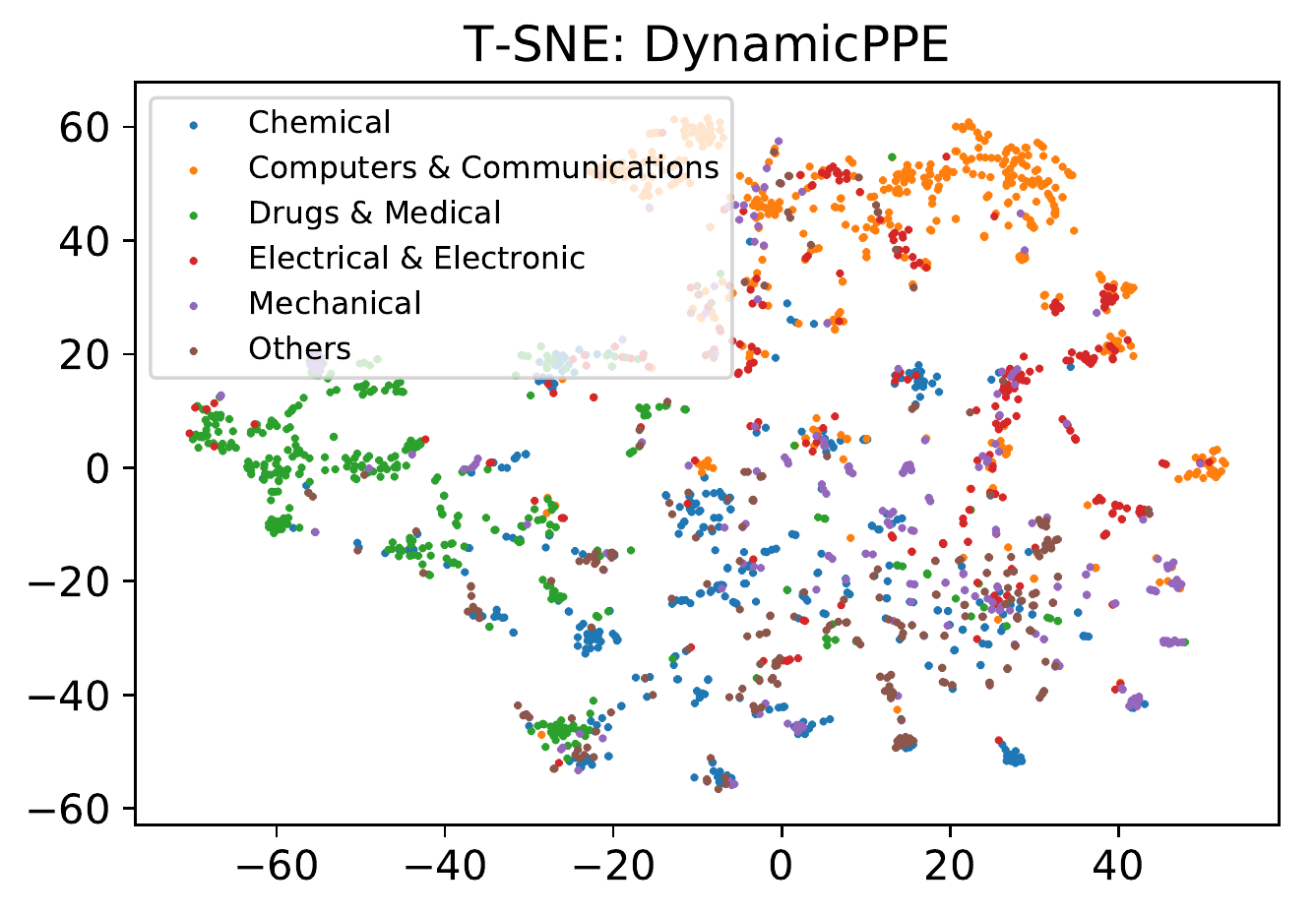}
	\end{subfigure}
    \caption{We randomly select 2,000 nodes from patent-small and visualize their embeddings using T-SNE\cite{van2008visualizing}}.
    \label{fig:nc-tsne-patent-small}
\end{figure}

\begin{figure}[htbp]
	\centering
	\begin{subfigure}{.23\textwidth}
		\includegraphics[width=3cm]{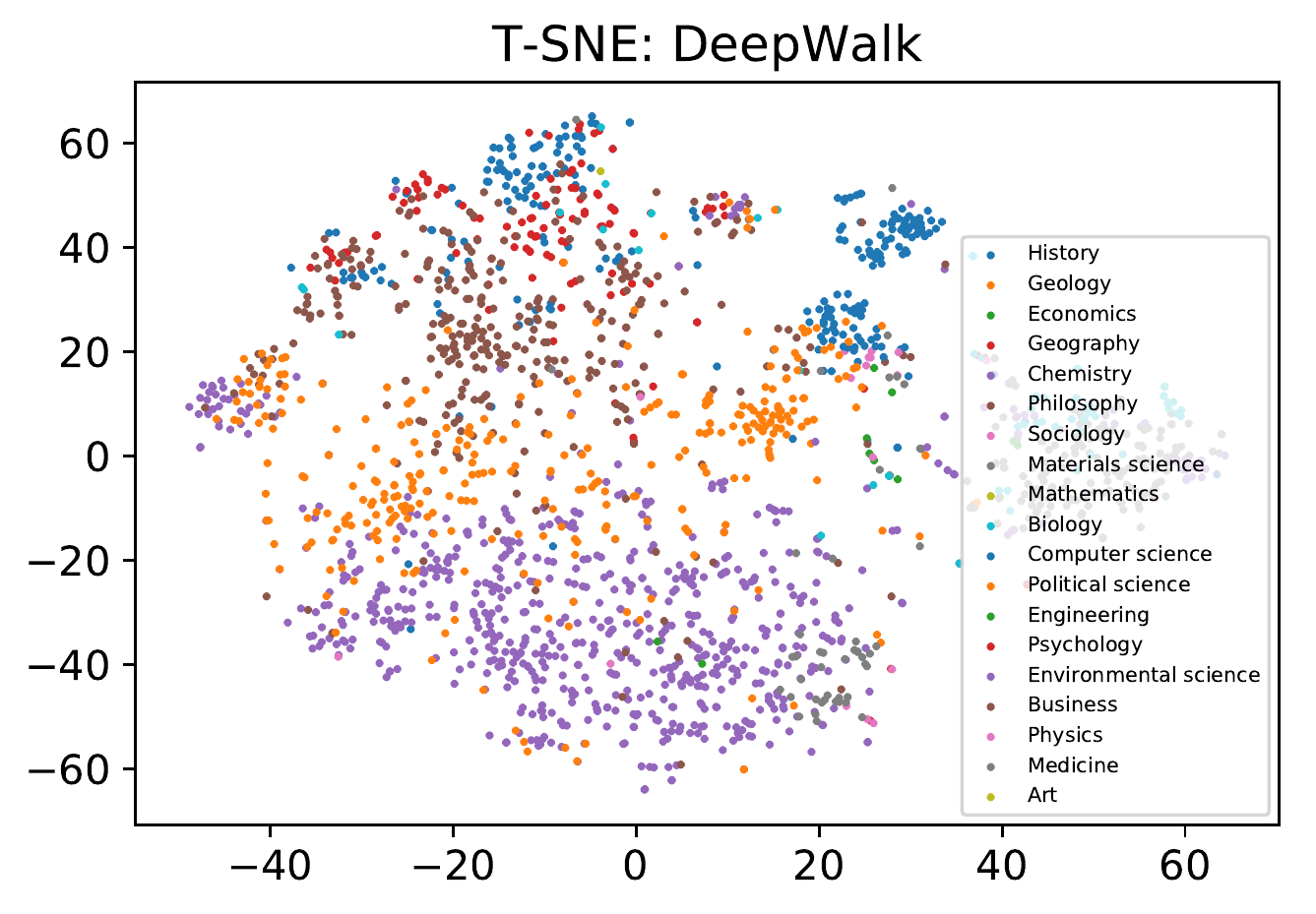}
	\end{subfigure}
	\begin{subfigure}{.23\textwidth}
		\includegraphics[width=3cm]{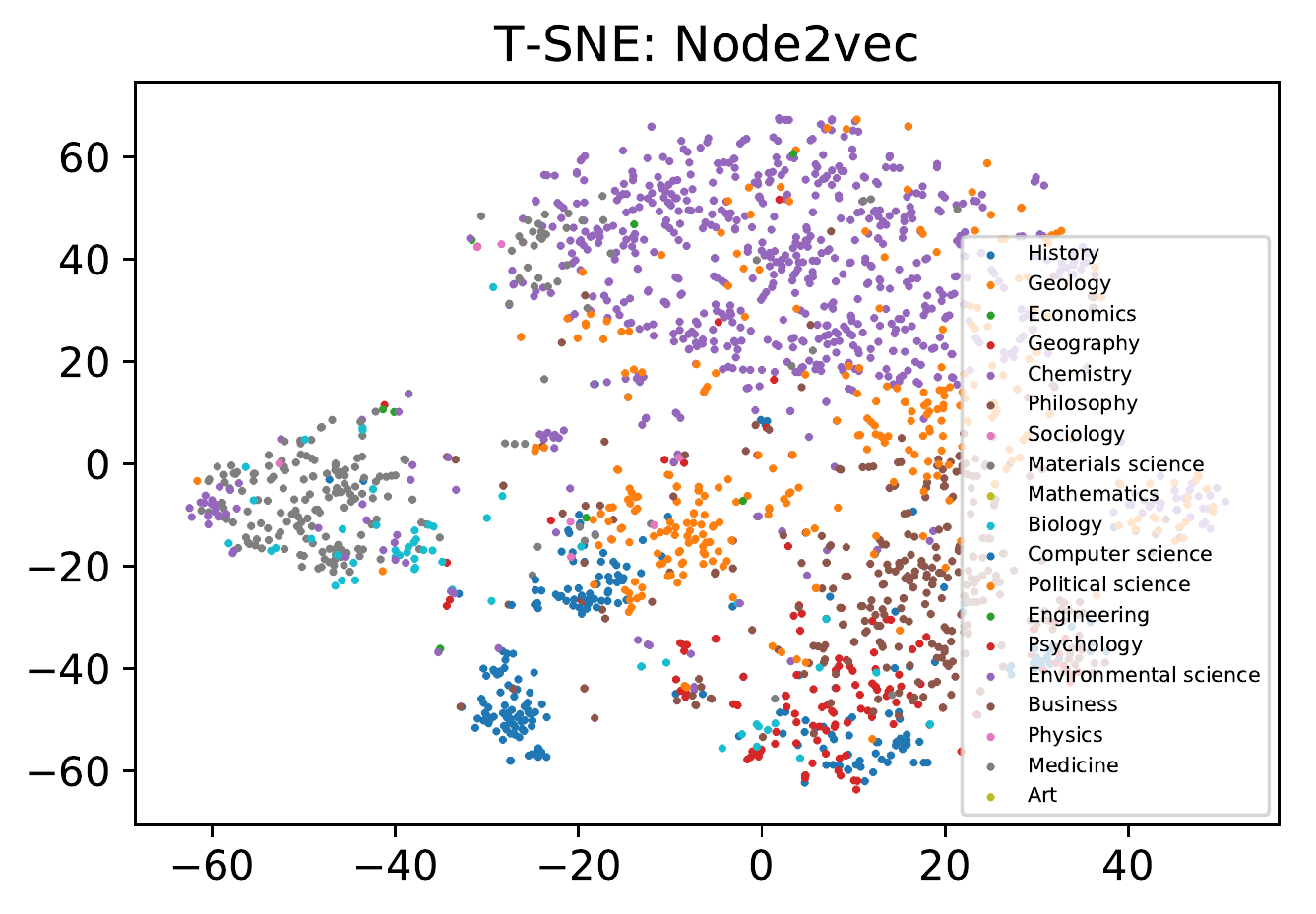}
	\end{subfigure}
	\begin{subfigure}{.23\textwidth}
		\includegraphics[width=3cm]{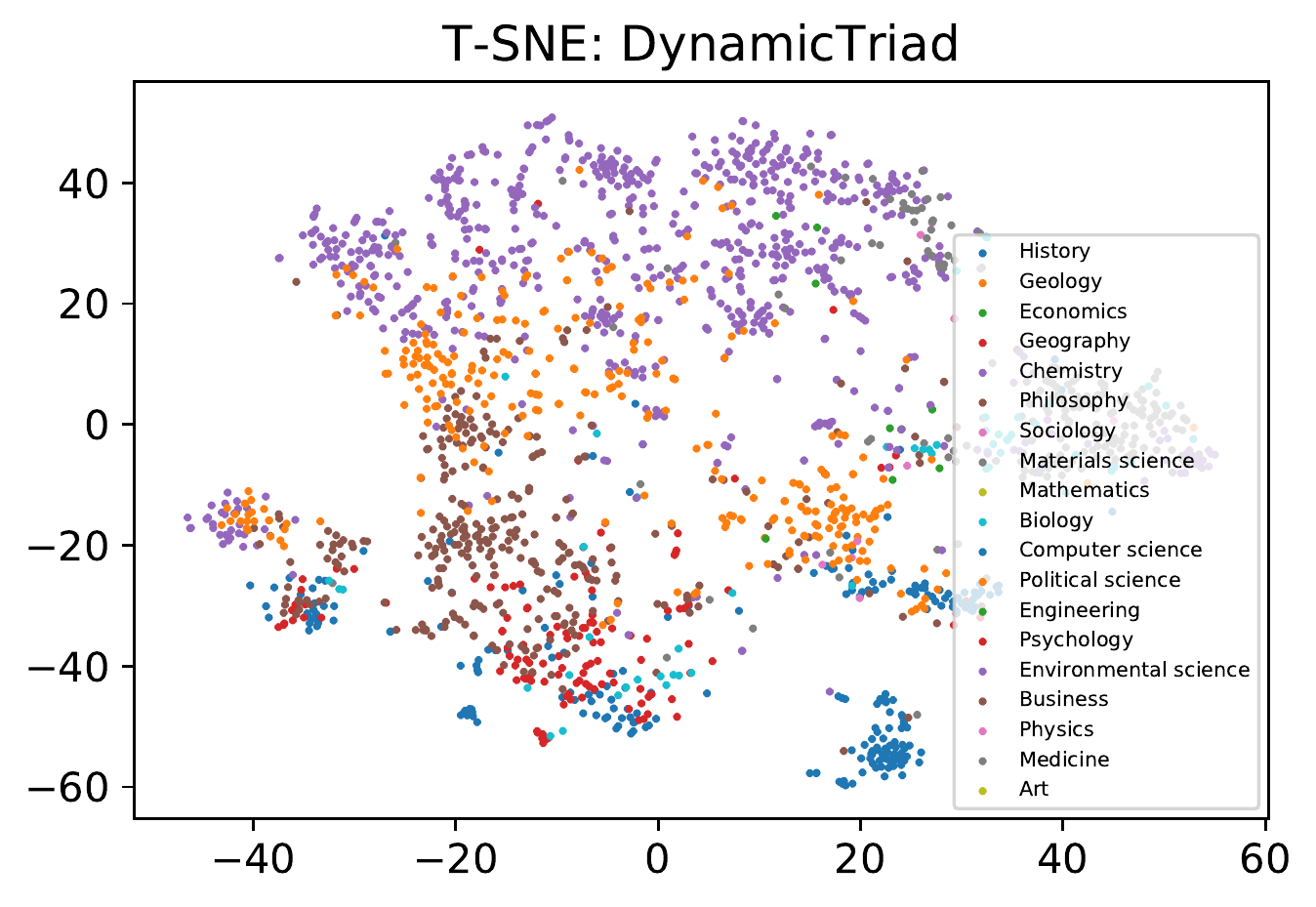}
	\end{subfigure}
	\begin{subfigure}{.23\textwidth}
		\includegraphics[width=3cm]{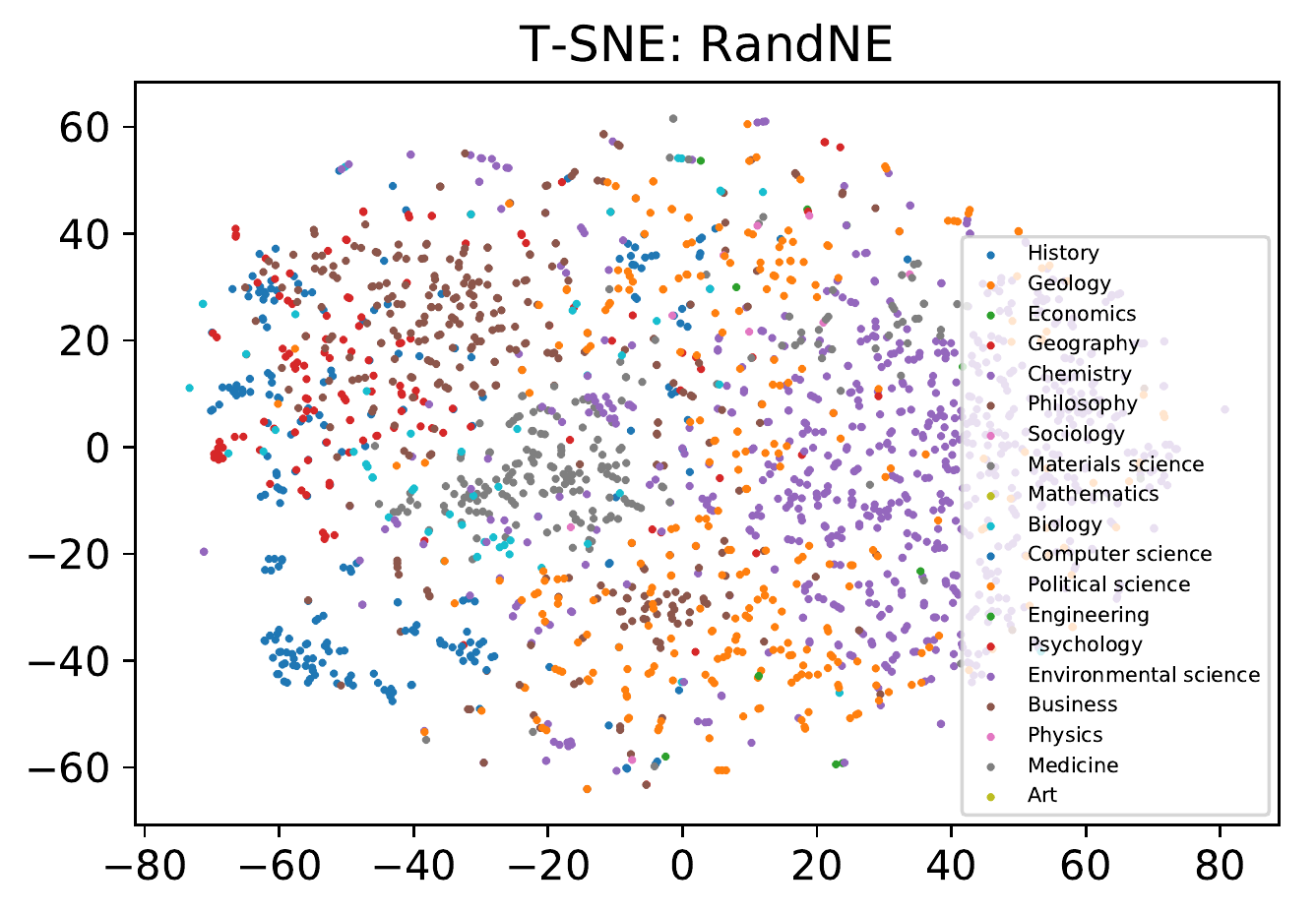}
	\end{subfigure}
	\begin{subfigure}{.23\textwidth}
		\includegraphics[width=3cm]{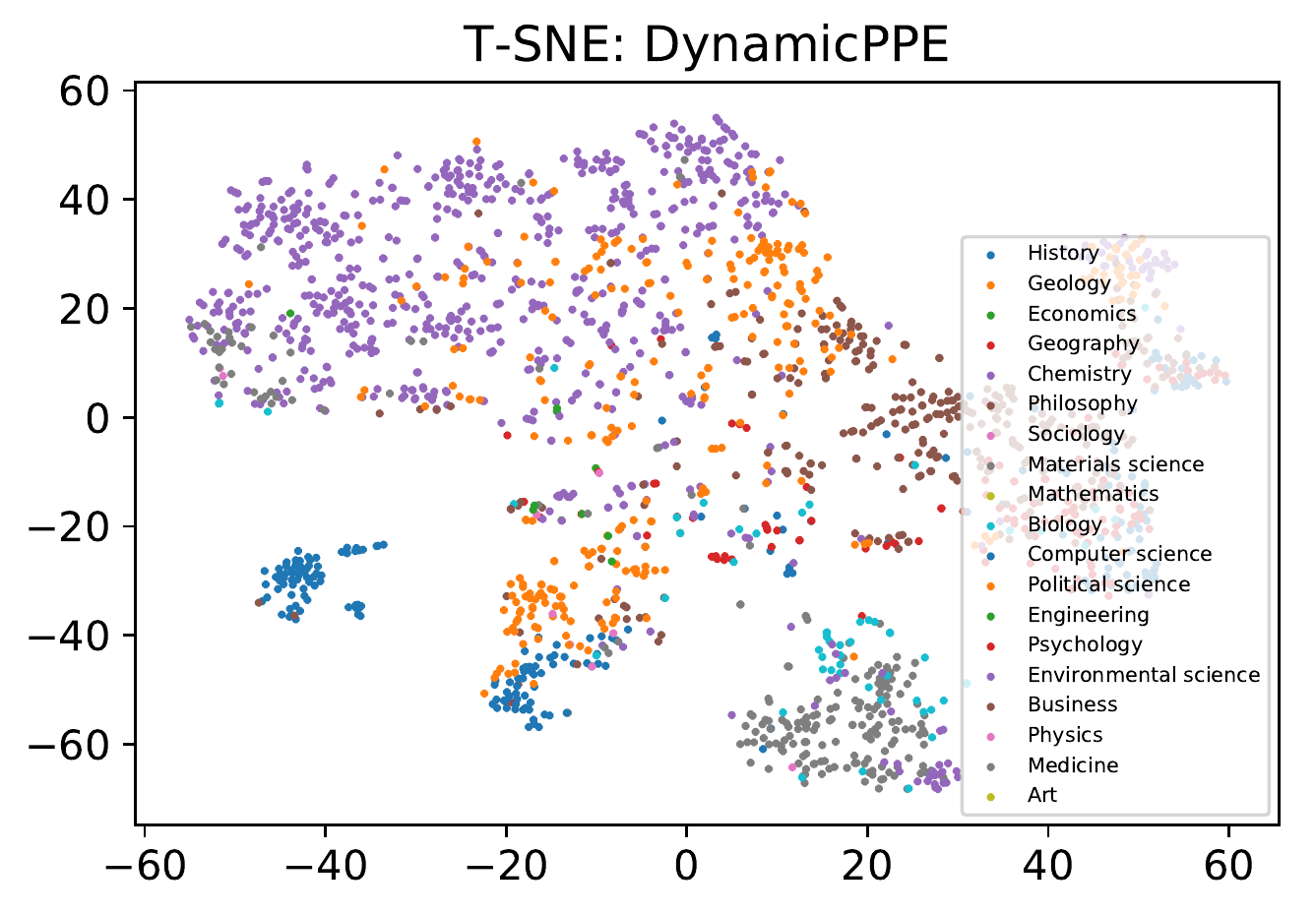}
	\end{subfigure}
    \caption{We randomly select 2,000 nodes from Coauthor-small graph and visualize their embeddings using T-SNE\cite{van2008visualizing}}.
    \label{fig:nc-tsne-coauthor-small}
\end{figure}

\begin{figure}[htbp]
	\centering
	\begin{subfigure}{.23\textwidth}
		\includegraphics[width=3cm]{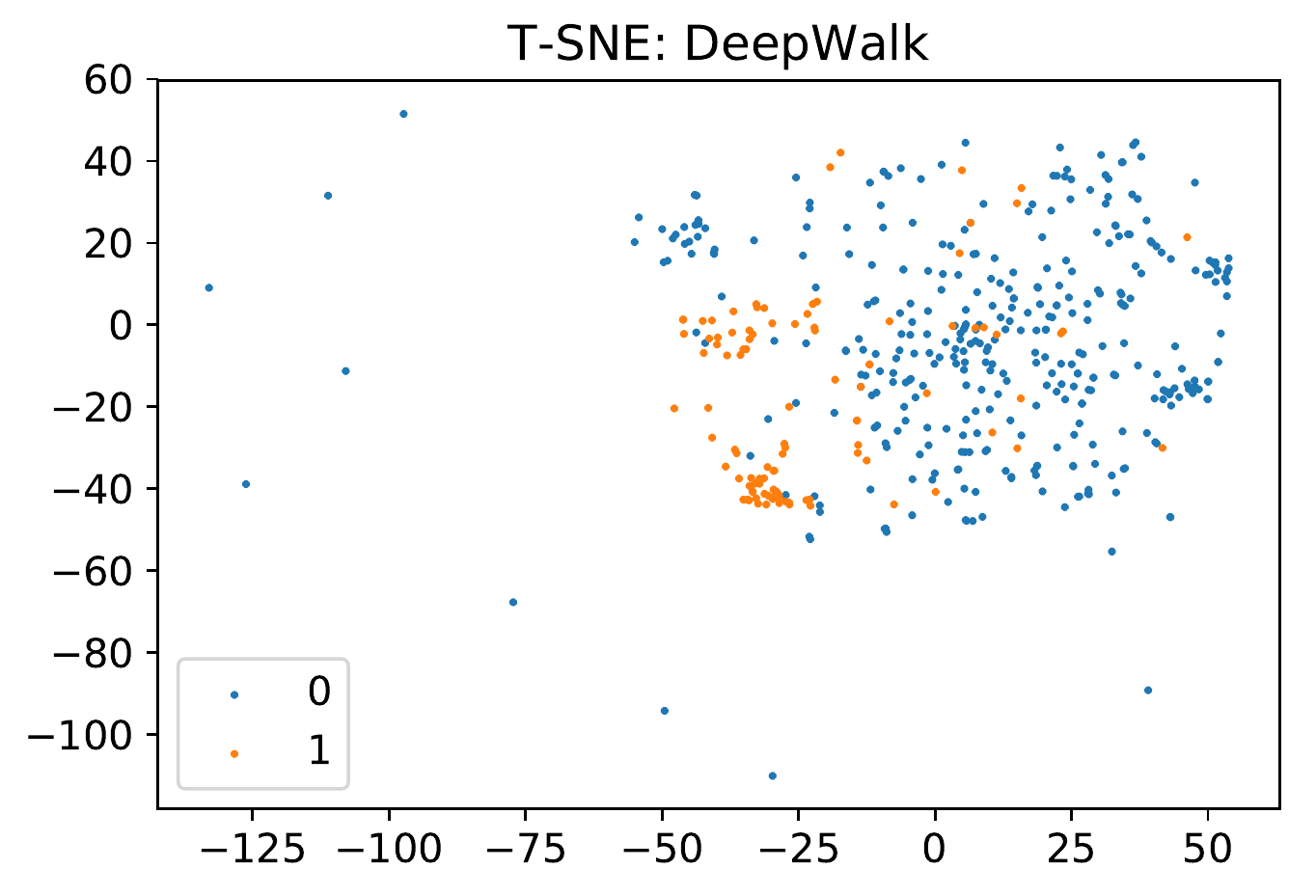}
	\end{subfigure}
	\begin{subfigure}{.23\textwidth}
		\includegraphics[width=3cm]{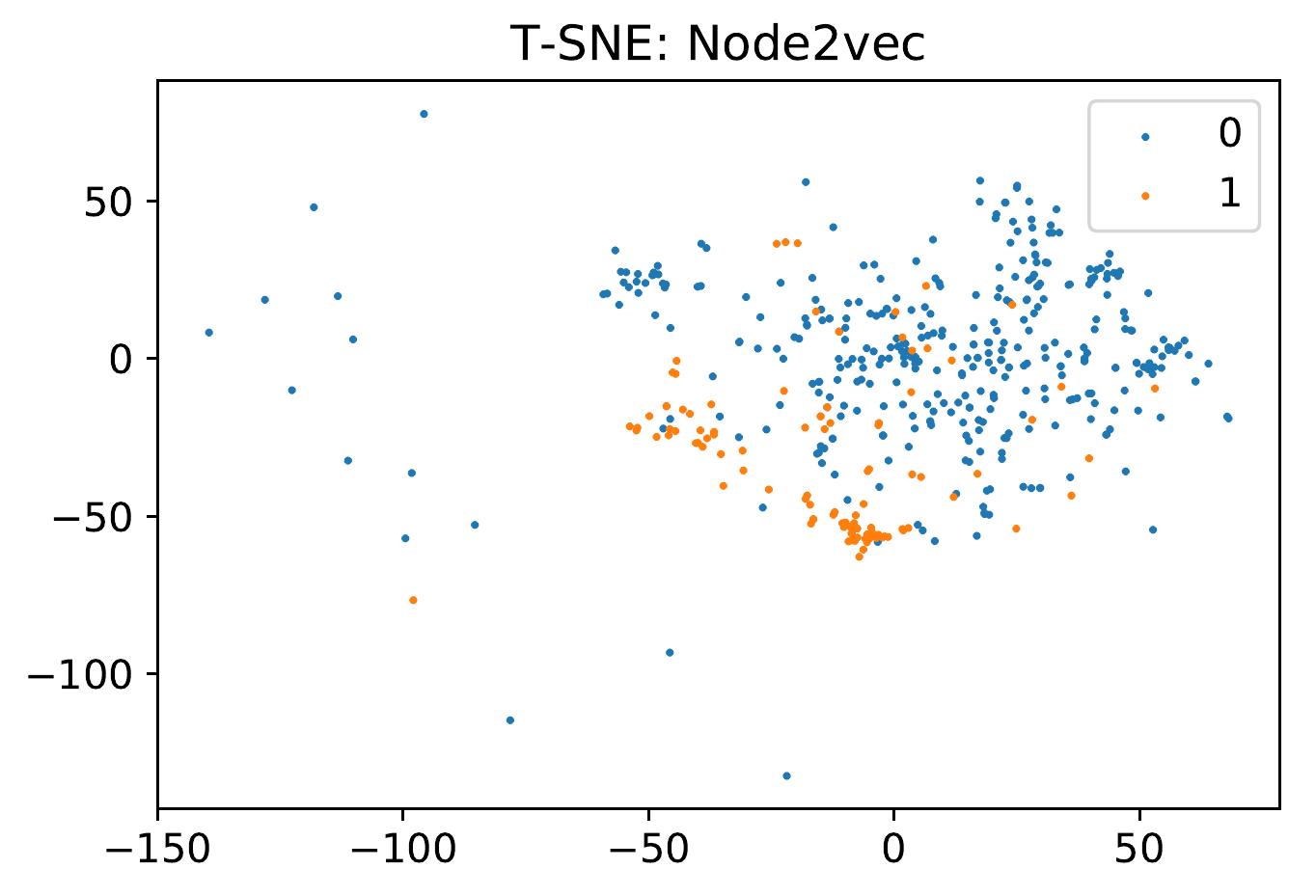}
	\end{subfigure}
	\begin{subfigure}{.23\textwidth}
		\includegraphics[width=3cm]{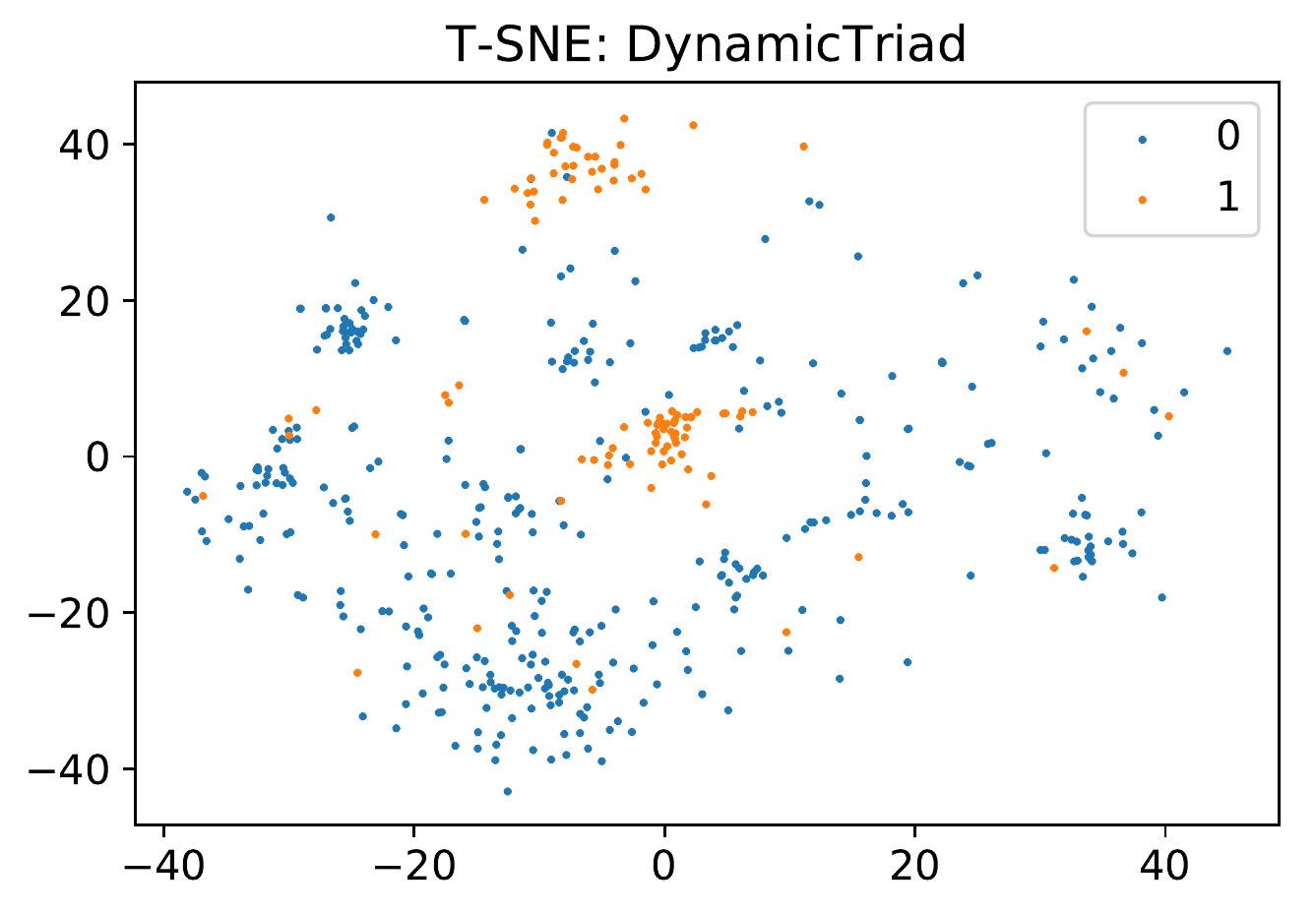}
	\end{subfigure}
	\begin{subfigure}{.23\textwidth}
		\includegraphics[width=3cm]{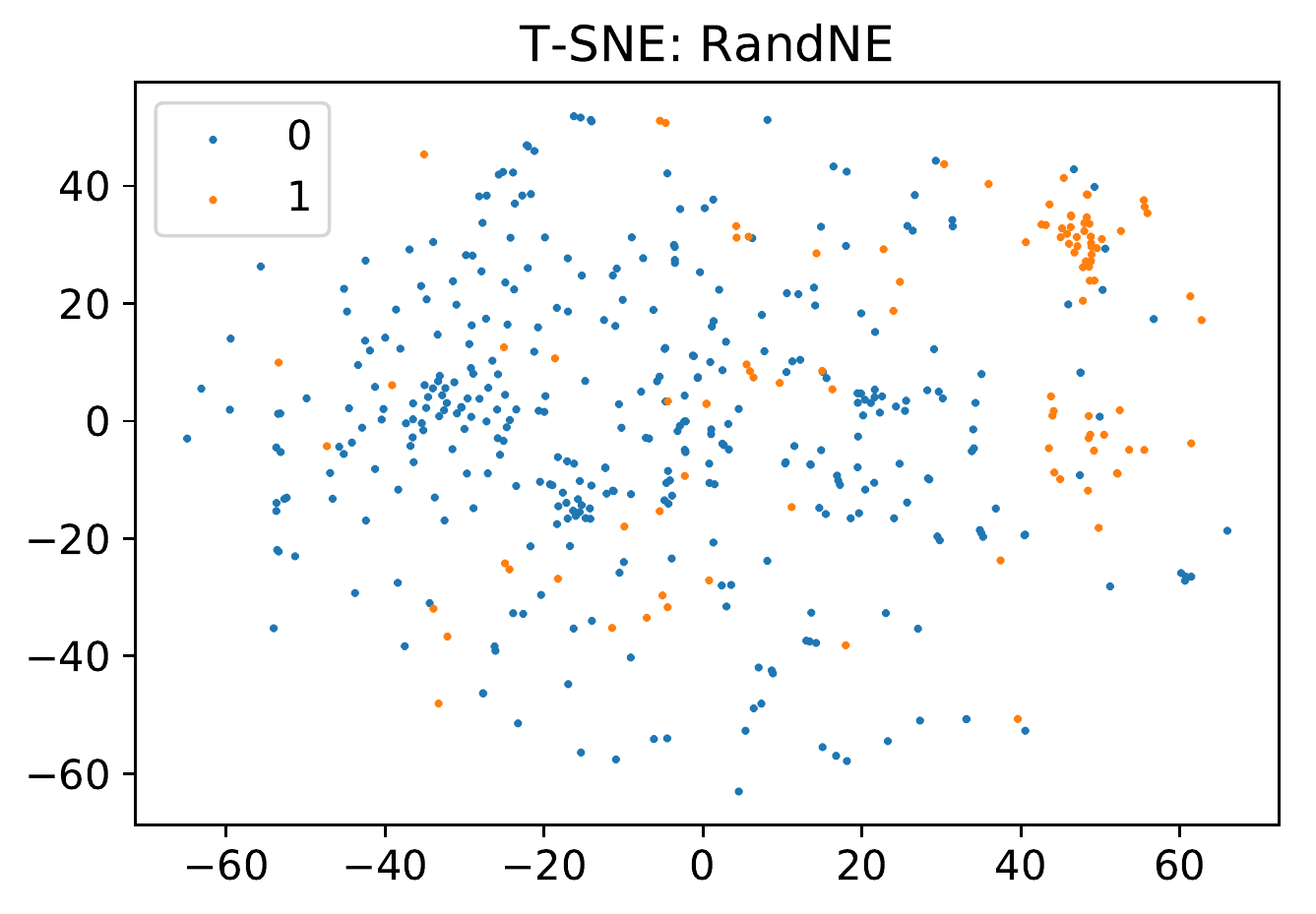}
	\end{subfigure}
	\begin{subfigure}{.23\textwidth}
		\includegraphics[width=3cm]{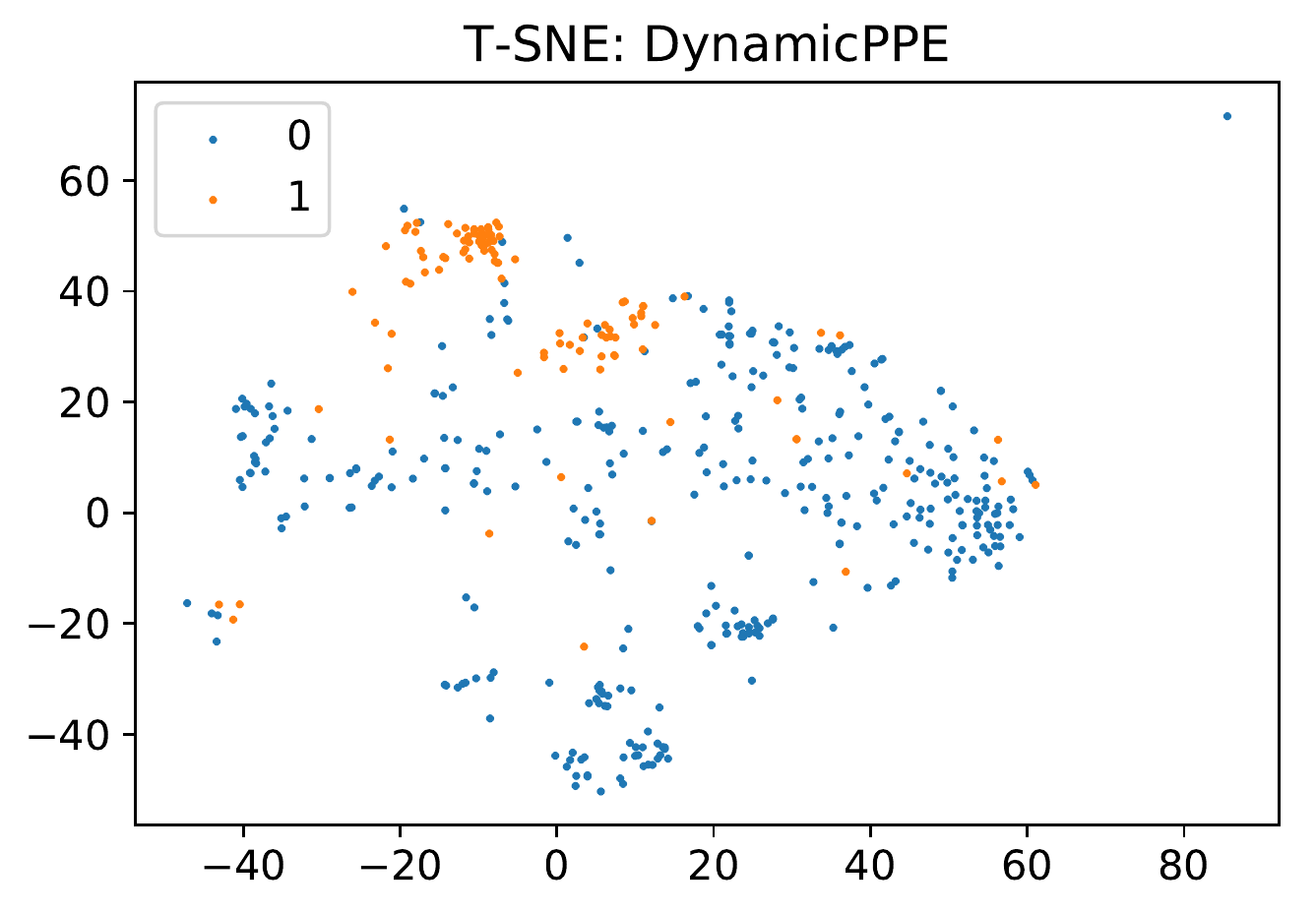}
	\end{subfigure}
    \caption{We select all labeled nodes from Academic graph and visualize their embeddings using T-SNE\cite{van2008visualizing}}.
    \label{fig:nc-tsne-academic-small}
\end{figure}

\end{document}